\documentclass[%
reprint,
amsmath,amssymb,
aps,
]{revtex4-1}

\usepackage{graphicx}
\usepackage{dcolumn}
\usepackage{bm}

\usepackage{subfigure} 

\usepackage{stmaryrd} 
\usepackage{amsthm}  
\usepackage{IEEEtrantools} 
\usepackage{dsfont} 

\usepackage{mathtools} 

\newtheorem{lem}{Lemma}
\newtheorem{theorem}{Theorem}

\newcommand{\Id}[1]{\boldsymbol{\mathbb{I #1}}}

\begin{document}

\preprint{APS/123-QED}

\title{Scattering matrix pole expansions for complex wavenumbers in $R$-matrix theory}

\author{Pablo Ducru}
\email{p\_ducru@mit.edu ; pablo.ducru@polytechnique.org}
\altaffiliation[]{Also from \'Ecole Polytechnique, France. \& Schwarzman Scholars, Tsinghua University, China.}
\affiliation{%
Massachusetts Institute of Technology\\
Department of Nuclear Science \& Engineering\\
77 Massachusetts Avenue, Cambridge, MA, 02139 U.S.A.\\
}%

\author{Vladimir Sobes}%
\email{sobesv@utk.edu}
\affiliation{%
University of Tennessee\\
Department of Nuclear Engineering\\
1412 Circle Drive, Knoxville, TN, 37996, U.S.A.
}%

\author{Gerald Hale}
\email{ghale@lanl.gov}
\author{Mark Paris}%
\email{mparis@lanl.gov}
\affiliation{%
Los Alamos National Laboratory\\
Theoretical Division (T-2)\\
MS B283, Los Alamos, NM 87454 U.S.A.\\
}%

\author{Benoit Forget}
\email{bforget@mit.edu}
\affiliation{%
Massachusetts Institute of Technology\\
Department of Nuclear Science \& Engineering\\
77 Massachusetts Avenue, Cambridge, MA, 02139 U.S.A.\\
}%


%


\date{\today}

\begin{abstract}
In this follow-up article to \cite{Ducru_shadow_Brune_Poles_2020}, we establish new results on scattering matrix pole expansions for complex wavenumbers in R-matrix theory.
In the past, two branches of theoretical formalisms emerged to describe the scattering matrix in nuclear physics: R-matrix theory, and pole expansions. The two have been quite isolated from one another.
Recently, our study of Brune's alternative parametrization of R-matrix theory has shown the need to extend the scattering matrix (and the underlying R-matrix operators) to complex wavenumbers \cite{Ducru_shadow_Brune_Poles_2020}.
Two competing ways of doing so have emerged from a historical ambiguity in the definitions of the shift $\boldsymbol{S}$ and penetration $\boldsymbol{P}$ functions: the legacy Lane \& Thomas ``force closure'' approach, versus analytic continuation (which is the standard in mathematical physics).
The R-matrix community has not yet come to a consensus as to which to adopt for evaluations in standard nuclear data libraries, such as ENDF \cite{ENDFBVIII8th2018brown}. 

In this article, we argue in favor of analytic continuation of R-matrix operators.
We bridge R-matrix theory with the Humblet-Rosenfeld pole expansions, and unveil new properties of the Siegert-Humblet radioactive poles and widths, including their invariance properties to changes in channel radii $a_c$.
We then show that analytic continuation of R-matrix operators preserves important physical and mathematical properties of the scattering matrix -- cancelling spurious poles and guaranteeing generalized unitarity -- while still being able to close channels below thresholds.

\end{abstract}

\maketitle


\section{\label{sec:Introduction}Introduction}

Myriad nuclear interactions have been modeled with R-matrix theory, with applications to many branches of nuclear physics, from nuclear simulation, radiation transport, astrophysics and cosmology, and extending to particle physics or atomistic and molecular simulation \cite{Blatt_and_Weisskopf_Theoretical_Nuclear_Physics_1952}\cite{Lane_and_Thomas_1958}\cite{atomistic_R_matrix_2011}\cite{ Computational_methods_in_R-matrix-2007, polyatomic-molecules-using-R-matrix-Royal-Society,  Burke-R-matrix-theory-of-atomic-processes, electron-molecule-scattering-R-matrix-Burke-1979, QB-method_R-matrix_1996}.
Our current nuclear data libraries are based on R-matrix evaluations (ENDF\cite{ENDFBVIII8th2018brown}, JEFF\cite{JEFF_2020_plompenJointEvaluatedFission2020}, BROND\cite{BROND_2016}, JENDL\cite{JENDL_shibataJENDL4NewLibrary2011}, CENDL\cite{CENDLProjectChinese2017}, TENDL\cite{TENDLkoningModernNuclearData2012, koningTENDLCompleteNuclear2019}). The R-matrix scattering model takes different incoming particle-waves and lets them interact through a given Hamiltonian to produce different possible outcomes.
R-matrix theory studies the particular two-body-in/two-body-out model of this scattering event, with the fundamental assumption that the total Hamiltonian is the superposition of a short-range, interior Hamiltonian, which is null after a given channel radius $a_c$, and a long-range, exterior Hamiltonian which is well known (free particles or Coulomb potential, for instance)\cite{Kapur_and_Peierls_1938, Wigner_and_Eisenbud_1947, Bloch_1957, Lane_and_Thomas_1958}.
R-matrix theory can parametrize the energy dependence of the scattering matrix in different ways.
The Wigner-Eisenbud parametrization is the historical standard, because its parameters are real and well defined, though some are arbitrary (the channel radius $a_c$ and the boundary condition $B_c$).
To remove this dependence on an arbitrary boundary condition $B_c$, the nuclear community has recently been considering shifting nuclear data libraries to an alternative parametrization of R-matrix theory \cite{Boundary_condition_Barker_1972, Brune_2002, anguloMatrixAnalysisInterference2000, Ducru_shadow_Brune_Poles_2020}.

In parallel, a vast literature in mathematical physics and nuclear physics has studied pole expansions of the scattering matrix \cite{Dyatlov, Guillope_1989_ENS, S-matrix-complex_resonance_parameters_Csoto_1997, Favella_and_Reineri_1962, complex_energy_above_inoization_threshold_1969}, starting with the well-known developments by Humblet and Rosenfeld \cite{Theory_of_Nuclear_Reactions_I_resonances_Humblet_and_Rosenfeld_1961, Theory_of_Nuclear_Reactions_II_optical_model_Rosenfeld_1961, Theory_of_Nuclear_Reactions_III_Channel_radii_Humblet_1961_channel_Radii, Theory_of_Nuclear_Reactions_IV_Coulomb_Humblet_1964, Theory_of_Nuclear_Reactions_V_low_energy_penetrations_Jeukenne_1965, Theory_of_Nuclear_Reactions_VI_unitarity_Humblet_1964, Theory_of_Nuclear_Reactions_VII_Photons_Mahaux_1965, Theory_of_Nuclear_Reactions_VIII_evolutions_Rosenfeld_1965}.

In this article, we show in section \ref{sec:Siegert-Humblet pole expansion in radioactive states} how the Siegert-Humblet expansion into radioactive states is the link between R-matrix theory and the scattering matrix pole expansions of Humblet \& Rosenfeld. In the process, we unveil new relations between the radioactive poles and residues and the alternative parametrization of R-matrix theory, and establish for the first time the number of radioactive poles in wavenumber space (theorem \ref{theo::Siegert-Humblet Radioactive Poles}) along with their branch-structure.
Section \ref{sec:Radioactive parameters invariance to channel radii} investigates the invariance properties of Siegert-Humblet radioactive parameters to a change in channel radius $a_c$. We demonstrate in theorem \ref{theo:: Radioactive parameters transformation under change of channel radius} that invariance of the scattering matrix to $a_c$ sets a partial differential equation on the Kapur-Peierls operator $\boldsymbol{R}_{L}$, which in turn enables us to derive explicit transformations of the Siegert-Humblet radioactive widths $\left\{r_{j,c}\right\}$ under a change of channel radius $a_c$.
Section \ref{sec::Scattering matrix continuation to complex energies} considers the continuation of the scattering matrix to complex wavenumbers in R-matrix theory. We establish several new results. We show that the legacy of Lane \& Thomas to force-close channels below threshold not only breaks the analytic properties of the scattering matrix, but also introduces nonphysical spurious poles. Yet we prove that these spurious poles are cancelled-out if one performs analytic continuation of R-matrix operators instead (theorem \ref{theo: Scattering matrix poles are the Siegert-Humblet radioactive poles}). 
We also show that this analytic continuation of R-matrix operators enforces the generalized unitarity conditions described by Eden \& Taylor \cite{Eden_and_Taylor} (theorem \ref{theo::satisfied generalized unitarity condition}).
Finally, in the case of massive particles, we propose a solution to the conundrum of how to close the channels below thresholds, by invoking both a quantum tunneling argument, whereby the transmission matrix is evanescent below threshold (theorem \ref{theo::evanescence of sub-threshold transmission matrix}), and a physical argument based on the definition of the cross section as the ratio of probability currents (theorem \ref{theo::Analytic continuation annuls sub-threshold cross sections}). 
All these results make us argue that, contrary to what Lane \& Thomas prescribed \cite{Lane_and_Thomas_1958}, the R-matrix parametrization should be analytically continued to complex wave-numbers $k_c\in\mathbb{C}$.
These considerations have practical implications on R-matrix evaluation codes, such as EDA \cite{EDA_2008, EDA_2015}, SAMMY \cite{SAMMY_2008}, or AZURE \cite{azumaAZUREMatrixCode2010}, used to build our nuclear data libraries (ENDF\cite{ENDFBVIII8th2018brown}, JEFF\cite{JEFF_2020_plompenJointEvaluatedFission2020}, BROND\cite{BROND_2016}, JENDL\cite{JENDL_shibataJENDL4NewLibrary2011}, CENDL\cite{CENDLProjectChinese2017}, TENDL\cite{TENDLkoningModernNuclearData2012, koningTENDLCompleteNuclear2019}). 
We thus call for analytic continuation of R-matrix operators to become the new standard for nuclear cross section evaluations.

\section{\label{sec:Siegert-Humblet pole expansion in radioactive states}Siegert-Humblet pole expansion in radioactive states}

We here establish new R-matrix theory results concerning the Siegert-Humblet expansion into radioactive states (c.f. sections IX.2.c-d-e p.297-298 in \cite{Lane_and_Thomas_1958}).
These \textit{radioactive parameters} express the energy dependence of the scattering matrix into a simple sum of poles and residues. We show this constitutes the link between R-matrix theory and the scattering matrix pole expansions of Humblet and Rosenfeld \cite{Theory_of_Nuclear_Reactions_I_resonances_Humblet_and_Rosenfeld_1961, Theory_of_Nuclear_Reactions_II_optical_model_Rosenfeld_1961, Theory_of_Nuclear_Reactions_III_Channel_radii_Humblet_1961_channel_Radii, Theory_of_Nuclear_Reactions_IV_Coulomb_Humblet_1964, Theory_of_Nuclear_Reactions_V_low_energy_penetrations_Jeukenne_1965, Theory_of_Nuclear_Reactions_VI_unitarity_Humblet_1964, Theory_of_Nuclear_Reactions_VII_Photons_Mahaux_1965, Theory_of_Nuclear_Reactions_VIII_evolutions_Rosenfeld_1965} (section \ref{sec: Radioactive parameters link R-matrix theory to the scattering matrix pole expansions}). In the wake, we show how to obtain the radioactive parameters (section \ref{sec:R_L def}), link them to the Brune alternative parametrization (section \ref{sec: Level matrix approach to Siegert-Humblet expansion}), reveal their branch structure (theorem \ref{theo::Siegert-Humblet Radioactive Poles} section \ref{sec: Siegert-Humblet Radioactive Pole Expansion Branch Structure}), which emerges from the wavenumber-energy mapping (\ref{eq:rho_c(E) mapping}):
\begin{equation}
\rho_c(E)  \quad \longleftrightarrow \quad E
\label{eq:rho_c(E) mapping}
\end{equation}
where $\rho_c \triangleq a_c k_c$ is the dimensionless wavenumber variable $\boldsymbol{\rho} \triangleq \boldsymbol{\mathrm{diag}}\left(\rho_c\right) $, $a_c$ is the arbitrary channel radius, and $k_c(E)$ is the wavenumber of channel $c$, which is linked to the energy $E$ of the system (eigenvalue of the Hamiltonian of the Sch\"odinger equation) as explained in section II.A. of \cite{Ducru_shadow_Brune_Poles_2020}.


\subsection{\label{sec:R_L def}Definition of Siegert \& Humblet parametrization}

Following the notation of \cite{Ducru_shadow_Brune_Poles_2020} (to which we refer for further explanations), we here recall the essential relation expressing the scattering matrix $\boldsymbol{U}$ as a function of R-matrix operators:
\begin{equation}
\begin{IEEEeqnarraybox}[][c]{rcl}
\boldsymbol{U} & \ = \ & \boldsymbol{O}^{-1}\boldsymbol{I} + 2\mathrm{i} \boldsymbol{\rho}^{1/2} \boldsymbol{O}^{-1} \boldsymbol{R}_{L} \boldsymbol{O}^{-1} \boldsymbol{\rho}^{1/2}
\IEEEstrut\end{IEEEeqnarraybox}
\label{eq:U expression}
\end{equation}
$\boldsymbol{I}$ and $\boldsymbol{O}$ are the incoming and outgoing wavefunctions, which are subject to the following Wronksian condition: for all channel $ c$, $ w_c \triangleq O_c^{(1)}I_c - I_c^{(1)}O_c = 2\mathrm{i} $, or with identity matrix $\Id{}$ (expression (7) in \cite{Ducru_shadow_Brune_Poles_2020}) and denoting $[ \; \cdot \; ]^{(1)}$ the diagonal channel $c$ derivative with respect to $\rho_c$
\begin{equation}
\begin{IEEEeqnarraybox}[][c]{rcl}
\boldsymbol{w} & \ \triangleq \ & \boldsymbol{O}^{(1)}\boldsymbol{I} - \boldsymbol{I}^{(1)}\boldsymbol{O} = 2\mathrm{i} \Id{}
\IEEEstrut\end{IEEEeqnarraybox}
\label{eq:wronksian expression}
\end{equation}
and where $\boldsymbol{R}_{L}$ is the \textit{Kapur-Peierls operator}, defined as (see equation (20) section II.D of \cite{Ducru_shadow_Brune_Poles_2020}):
\begin{equation}
\begin{IEEEeqnarraybox}[][c]{rcl}
\boldsymbol{R}_{L}  & \ \triangleq  \ &  \left[ \Id{} -  \boldsymbol{R}\left(\boldsymbol{L} - \boldsymbol{B} \right) \right]^{-1} \boldsymbol{R}  =   \boldsymbol{\gamma}^\mathsf{T} \boldsymbol{A} \boldsymbol{\gamma}
\label{eq: def Kapur-Peierls operator}
\IEEEstrut\end{IEEEeqnarraybox}
\end{equation}
This Kapur-Peierls $\boldsymbol{R}_{L}$ operator is at the heart of the Siegert-Humblet parametrization, and its study composes a core part of this article. 
The Kapur-Peierls operator $\boldsymbol{R}_{L}$ is a function of the Wigner-Eisenbud R-matrix $\boldsymbol{R}$, parametrized by the \textit{resonance parameters} (energies $\boldsymbol{e} \triangleq \boldsymbol{\mathrm{diag}}\left(E_\lambda \right)$ and widths $\boldsymbol{\gamma} \triangleq \boldsymbol{\mathrm{mat}}\left(\gamma_{\lambda,c}\right)$)
\begin{equation}
\begin{IEEEeqnarraybox}[][c]{rcl}
 \boldsymbol{R}(E) & \; \triangleq \; & \boldsymbol{\gamma}^\mathsf{T} \left(\boldsymbol{e} - E\; \Id{}\right)^{-1}\boldsymbol{\gamma}
\IEEEstrut\end{IEEEeqnarraybox}
\label{eq:R expression}
\end{equation}
as well as the arbitrary boundary condition $\boldsymbol{B} \triangleq \boldsymbol{\mathrm{diag}}\left(B_c\right) $, and the reduced logarithmic derivative of the outgoing wavefunction $\boldsymbol{L} \triangleq \boldsymbol{\mathrm{diag}}\left(L_c\right) $, defined as (c.f. section II.B of \cite{Ducru_shadow_Brune_Poles_2020})
\begin{equation}
    L_c(\rho_c)   \triangleq  \frac{\rho_c}{O_c} \frac{\partial O_c}{\partial \rho_c}
    \label{eq: L operator}
\end{equation}
and which admits the following Mittag-Leffler pole expansion (theorem 1 section II.B of \cite{Ducru_shadow_Brune_Poles_2020}):
\begin{equation}
\begin{IEEEeqnarraybox}[][c]{rcl}
    \frac{L_c(\rho)}{\rho} & = & \frac{- \ell}{\rho} + \mathrm{i} + \sum_{n \geq 1}  \frac{1}{\rho-\omega_n}
\IEEEstrut\end{IEEEeqnarraybox}
\label{eq::Explicit Mittag-Leffler expansion of L_c}
\end{equation}
where $\left\{\omega_n \right\}$ are the roots of the $O_c(\rho)$ outgoing wavefunctions: $\forall n , \; O_c(\omega_n) = 0$ (reported in table II of \cite{Ducru_shadow_Brune_Poles_2020} for neutral particles). \\
Equivalently, the Kapur-Peierls operator $\boldsymbol{R}_{L}$ can be expressed with the level matrix $\boldsymbol{A}$  (see equations (17) and (18) of section II.C of \cite{Ducru_shadow_Brune_Poles_2020}):
\begin{equation}
\begin{IEEEeqnarraybox}[][c]{rcl}
\boldsymbol{A^{-1}} & \ \triangleq \ & \boldsymbol{e} - E\Id{} - \boldsymbol{\gamma}\left( \boldsymbol{L} - \boldsymbol{B} \right)\boldsymbol{\gamma}^\mathsf{T}
\IEEEstrut\end{IEEEeqnarraybox}
\label{eq:inv_A expression}
\end{equation}

All these R-matrix operators are functions of the wavenumbers $k_c(E)$ (or their corresponding dimensionless wavenumber variable $\boldsymbol{\rho} \triangleq \boldsymbol{\mathrm{diag}}\left(\rho_c\right) $).
The Siegert-Humblet pole expansion in radioactive states consists of analytically continuing the Kapur-Peierls $\boldsymbol{R}_{L}$ operator to complex wavenumbers $k_c \in \mathbb{C}$, thereby becoming a locally meromorphic operator.
The poles of this meromorphic operator can be assumed to have a Laurent expansion of order one (i.e. simple poles), as we will discuss in section \ref{subsubsec:: Semi-simple poles in R-matrix theory}.
Since the Kapur-Peierls $\boldsymbol{R}_{L}$ operator is complex-symmetric, its residues at any given pole value $\mathcal{E}_j \in \mathbb{C}$ are also complex-symmetric. For non-degenerate eigenvalues $\mathcal{E}_j \in \mathbb{C}$, the corresponding residues are rank-one and expressed as $\boldsymbol{r_j}\boldsymbol{r_j}^\mathsf{T}$, while for degenerate eigenvalues $\mathcal{E}_j \in \mathbb{C}$ of multiplicity $M_j$, the corresponding residues are rank-$M_j$ and expressed as $ \sum_{m=1}^{M_j}\boldsymbol{r_{j}^m}{\boldsymbol{r_{j}^m}}^\mathsf{T}$.
On a given domain, the Mittag-Leffler theorem \cite{Mittag-Leffler_1884, Pacific_Journal_Mittag_Leffler_and_spectral_theory_1960} then states that $\boldsymbol{R}_{L}$ locally takes the form, in the vicinity $\mathcal{W}(E)$ (neighborhood) of any complex energy $E\in\mathbb{C}$ away from the branch points (threshold energies $E_{T_c}$) of mapping (\ref{eq:rho_c(E) mapping}), of a sum of poles and residues and a holomorphic entire part $\boldsymbol{\mathrm{Hol}}_{\boldsymbol{R}_{L}}(E)$:
\begin{equation}
\boldsymbol{R}_{L}(E) \underset{\mathcal{W}(E)}{=} \sum_{j\geq 1} \frac{\sum_{m=1}^{M_j}\boldsymbol{r_{j}^{m}}{\boldsymbol{{r_{j}^{m}}}}^\mathsf{T}}{E - \mathcal{E}_j} + \boldsymbol{\mathrm{Hol}}_{\boldsymbol{R}_{L}}(E)
\label{eq::RL Mittag Leffler degenerate state}
\end{equation}
or, in the particular (but most common) case where $\mathcal{E}_j$ is a non-degenerate eigenvalue (with multiplicity $M_j=1)$,
\begin{equation}
\boldsymbol{R}_{L}(E) \underset{\mathcal{W}(E)}{=} \sum_{j\geq 1} \frac{\boldsymbol{r_j}\boldsymbol{r_j}^\mathsf{T}}{E - \mathcal{E}_j} + \boldsymbol{\mathrm{Hol}}_{\boldsymbol{R}_{L}}(E)
\label{eq::RL Mittag Leffler}
\end{equation}
This is the Siegert-Humblet expansion into so-called \textit{radioactive states} \cite{Siegert_1939,Breit_radioactive_1940, Radioactive_Mahaux_1969, Eigenchannel_Treatment_Of_R_Matrix_Theory_1997,Radioactive_states_Roumania_2012}
--- equivalent to equation (2.16) of section IX.2.c. in \cite{Lane_and_Thomas_1958} where we have modified the notation for greater consistency ($\mathcal{E}_j$ corresponds to $H_\lambda$ of \cite{Lane_and_Thomas_1958} and $\boldsymbol{r_j}$ corresponds to $\boldsymbol{\omega_\lambda}$) since there are more complex poles $\mathcal{E}_j$ than real energy levels $E_\lambda$. The Siegert-Humblet parameters are then the poles $\left\{\mathcal{E}_j\right\}$ and residue widths $\left\{\boldsymbol{r_j}\right\}$ of this complex resonance expansion of the Kapur-Peierls operator $\boldsymbol{R}_{L}$.

The Gohberg-Sigal theory \cite{Gohberg_Sigal_1971} provides a method for calculating these poles and residues by solving the following generalized eigenvalue problem --- which we call \textit{radioactive problem}:
\begin{equation}
\left.\boldsymbol{R}_{L}^{-1}(E)\right|_{E = \mathcal{E}_j} \boldsymbol{q_j} = \boldsymbol{0}
\label{eq:R_L radioactive problem}
\end{equation}
that is finding the poles $\big\{\mathcal{E}_j\big\}$ of the Kapur-Peierls operator, $\boldsymbol{R}_{L}$, and their associated eigenvectors $\left\{\boldsymbol{q_j}\right\}$. The poles are complex and usually decomposed as:
\begin{equation}
\mathcal{E}_j \triangleq E_j - \mathrm{i}\frac{\Gamma_j}{2}
\label{eq:E_j pole def}
\end{equation}
It can be shown (c.f. discussion section IX.2.d pp.297--298 in \cite{Lane_and_Thomas_1958}, or section 9.2 eq. (9.11) in \cite{Theory_of_Nuclear_Reactions_I_resonances_Humblet_and_Rosenfeld_1961}) that fundamental physical properties (conservation of probability, causality and time reversal) ensure that the poles reside either on the positive semi-axis of purely-imaginary $k_c \in \mathrm{i}\mathbb{R}_+$ -- corresponding to bound states for real sub-threshold energies, i.e. $E_j < E_{T_c} $ and $\Gamma_j = 0$ -- or that all the other poles are on the lower-half $k_c$ plane, with $\Gamma_j > 0$, corresponding to ``resonance'' or ``radioactively decaying'' states. All poles enjoy the specular symmetry property: if $k_c \in \mathbb{C}$ is a pole of the Kapur-Peierls operator, then $-k_c^*$ is too.

Let $M_j = \mathrm{dim}\left( \mathrm{Ker} \left( \boldsymbol{R}_{L}^{-1}(\mathcal{E}_j) \right) \right) $ be the dimension of the nullspace of the inverse of the Kapur-Peierls operator at pole value $\mathcal{E}_j$ -- that is $M_j$ is the geometric multiplicity.
We can thus write $\mathrm{Ker} \left( \boldsymbol{R}_{L}^{-1}(\mathcal{E}_j) \right) = \mathrm{Span}\left( \boldsymbol{q_{j}^{1}}, \hdots, \boldsymbol{q_{j}^{m}},\hdots, \boldsymbol{q_{j}^{M_j}} \right)$.
As we discuss in section \ref{subsubsec:: Semi-simple poles in R-matrix theory}, it is physically reasonable to assume that the geometric and algebraic multiplicities are equal (semi-simplicity condition), which entails a Laurent development of order one for the poles -- i.e. no higher powers of $\frac{1}{E - \mathcal{E}_j}$ in expansion (\ref{eq::RL Mittag Leffler degenerate state}).
Since $\boldsymbol{R}_{L}$ is complex-symmetric, if we assume we can find non-quasi-null eigenvectors solutions to (\ref{eq:R_L radioactive problem}) -- that is $\forall \; (j,m) \; , \; \; {\boldsymbol{q_{j}^{m}}}^\mathsf{T} \boldsymbol{q_{j}^{m}} \neq 0 $ so it is non-defective \cite{Craven_complex_symmetric_1969, Nondefective_complex_symmetric_matrices_1985, complex_symmetric_matrix_SVD_1988, Scott_complex_symmetric_1993, fast_diag_of_complex_symmetric_matrices_for_quantum_applications_1997, Complex_symmetric_operators_2005, Complex_symmetric_operators_II_2007} -- then Gohberg-Sigal theory can be adapted to the case of complex-symmetric matrices to normalize the rank-$M_j$ residues of $\boldsymbol{R}_{L}$ matrix as:
\begin{equation}
\sum_{m=1}^{M_j} \boldsymbol{r_{j}^{m}} {\boldsymbol{r_{j}^{m}}}^\mathsf{T}  =  \sum_{m=1}^{M_j} \frac{\boldsymbol{q_{j}^{m}}{\boldsymbol{q_{j}^{m}}}^\mathsf{T}}{{\boldsymbol{q_{j}^{m}}}^\mathsf{T} \left( \left. { \frac{\partial \boldsymbol{R}_{L}^{-1}}{\partial E} }\right|_{E=\mathcal{E}_j} \right) {\boldsymbol{q_{j}^{m}}}}
\label{eq:R_L residues for degenerate case}
\end{equation}
The residue widths $\left\{\boldsymbol{r_{j}^{m}}\right\}$, here called \textit{radioactive widths}, can thus directly be expressed as:
\begin{equation}
\boldsymbol{r_{j}^{m}}  = \frac{\boldsymbol{q_{j}^{m}}}{\sqrt{{\boldsymbol{q_{j}^{m}}}^\mathsf{T} \left( \left. { \frac{\partial \boldsymbol{R}_{L}^{-1}}{\partial E} }\right|_{E=\mathcal{E}_j} \right) \boldsymbol{q_{j}^{m}}}}
\label{eq:R_L residues widths}
\end{equation}
where $\left. { \frac{\partial \boldsymbol{R}_{L}^{-1}}{\partial E} }\right|_{E=\mathcal{E}_j} $ can readily be calculated by means of the following property:
\begin{equation}
\left. { \frac{\partial \boldsymbol{R}_{L}^{-1}}{\partial E} }\right|_{E=\mathcal{E}_j} = \frac{\partial \boldsymbol{R}^{-1} }{\partial E} (\mathcal{E}_j) - \frac{\partial \boldsymbol{L} }{\partial E} (\mathcal{E}_j)
\end{equation}
where the R-matrix $\boldsymbol{R}$ is invertible at the radioactive poles $\left\{\mathcal{E}_j\right\}$, with
\begin{equation}
\frac{\partial \boldsymbol{R}^{-1} }{\partial E} (E) = - \boldsymbol{R}^{-1} \boldsymbol{\gamma}^\mathsf{T} \left(\boldsymbol{e} - E\Id{}\right)^{-2} \boldsymbol{\gamma} \boldsymbol{R}^{-1}
\end{equation}
In practice, we are most often presented with non-degenerate states where $M_j = 1$, meaning the kernel is an eigenline $\mathrm{Ker} \left( \boldsymbol{R}_{L}^{-1}(\mathcal{E}_j) \right) = \mathrm{Span}\left( \boldsymbol{q_{j}} \right)$, which entails rank-one residues normalized as:
\begin{equation}
\boldsymbol{r_j} \boldsymbol{r_j}^\mathsf{T}  = \frac{\boldsymbol{q_j}\boldsymbol{q_j}^\mathsf{T}}{\boldsymbol{q_j}^\mathsf{T} \left( \left. { \frac{\partial \boldsymbol{R}_{L}^{-1}}{\partial E} }\right|_{E=\mathcal{E}_j} \right) \boldsymbol{q_j}}
\label{eq:R_L residues}
\end{equation}
or equivalently
\begin{equation}
\boldsymbol{r_j}^\mathsf{T} \left( \left. { \frac{\partial \boldsymbol{R}_{L}^{-1}}{\partial E} }\right|_{E=\mathcal{E}_j} \right) \boldsymbol{r_j} = 1 
\label{eq:R_L residues normalization}
\end{equation}
thus for clarity of reading and without loss of generality, we henceforth drop the superscript ``$m$'' and summation over the multiplicity, unless it is of specific interest.

The \textit{radioactive poles}, $\left\{\mathcal{E}_j\right\}$, and \textit{radioactive widths}, $\left\{\boldsymbol{r_{j}^{m}} = \left[ r_{{j,c_1}}^{m}, \hdots, r_{{j,c}}^m , \hdots , r_{{j,c_{N_c}}}^{m} \right]^\mathsf{T} \right\}$, are the Siegert-Humblet parameters. They are complex and locally untangle the energy dependence into an expansion sum of poles and residues (\ref{eq::RL Mittag Leffler degenerate state}).
Additional discussion on these poles and residues can be found in \cite{Lane_and_Thomas_1958}, sections IX.2.c-d-e p.297-298, or in \cite{Siegert, Breit_radioactive_1940, Radioactive_Mahaux_1969, Radioactive_states_Roumania}.

The Kapur-Peierls matrix $\boldsymbol{R}_{L}$ is invariant to a change in boundary conditions $B_c$ --- c.f. equations (25) and (26) of section II.F of \cite{Ducru_shadow_Brune_Poles_2020} --- this entails the radioactive poles $\left\{\mathcal{E}_j\right\}$ and widths $\left\{\boldsymbol{r_j}\right\}$ are independent of the boundary condition $B_c$.

\subsection{\label{sec: Level matrix approach to Siegert-Humblet expansion}Level matrix $\boldsymbol{A}(E)$ approach to Siegert \& Humblet expansion}

An alternative approach to calculating the Siegert-Humblet parameters $\left\{a_c,\mathcal{E}_j,r_{j,c}^m,E_{T_c}\right\}$ from the Wigner-Eisenbud ones $\left\{a_c, B_c, \gamma_{\lambda,c}, E_\lambda, E_{T_c} \right\}$ is through the level matrix $\boldsymbol{A}$.
We search for the poles and eigenvectors of the level matrix operator $\boldsymbol{A}$:
\begin{equation}
\left.\boldsymbol{A}^{-1}(E)\right|_{E=\mathcal{E}_j}\boldsymbol{b_j} = \boldsymbol{0}
\label{eq::invA det roots}
\end{equation}
from (\ref{eq:inv_A expression}), this means solving for the eigenvalues $\left\{\mathcal{E}_j\right\}$ and associated eigenvectors $\left\{\boldsymbol{b_j}\right\}$ that satisfy:
\begin{equation}
\begin{IEEEeqnarraybox}[][c]{rCl}
\left[\boldsymbol{e} - \boldsymbol{\gamma} \left( \boldsymbol{L}(\mathcal{E}_j) - \boldsymbol{B} \right)\boldsymbol{ \gamma}^\mathsf{T} \right]\boldsymbol{b_j} = \mathcal{E}_j\boldsymbol{b_j}
\IEEEstrut\end{IEEEeqnarraybox}
\label{eq: invA eigenproblem for RL poles}
\end{equation}
This problem is analogous to the alternative parametrization of R-matrix theory, but replacing the shift factor $\boldsymbol{S}$ with the outgoing-wave reduced logarithmic derivative $\boldsymbol{L}$ (c.f.  \cite{Ducru_shadow_Brune_Poles_2020}).

Again, the same hypotheses as for the Kapur-Peierls operator $\boldsymbol{R}_{L}$ in section \ref{sec:R_L def} allow us to adapt the Gohberg-Sigal theory to the case of complex-symmetric operators to yield the following local Mittag-Leffler expansion of the level matrix (with normalized residues):
\begin{equation}
\boldsymbol{A}(E) \underset{\mathcal{W}(E)}{=} \sum_{j\geq 1} \frac{ \sum_{m=1}^{M_j}\boldsymbol{a_{j}^{m}}{\boldsymbol{a_{j}^{m}}}^\mathsf{T}}{E - \mathcal{E}_j} + \boldsymbol{\mathrm{Hol}}_{\boldsymbol{A}}(E)
\label{eq::A Mittag Leffler}
\end{equation}
In the most frequent case of non-degenerate eigenvalues to (\ref{eq::invA det roots}), this yields rank-one residues as:
\begin{equation}
\boldsymbol{A}(E) \underset{\mathcal{W}(E)}{=} \sum_{j\geq 1} \frac{ \boldsymbol{a_j}\boldsymbol{a_j}^\mathsf{T}}{E - \mathcal{E}_j} + \boldsymbol{\mathrm{Hol}}_{\boldsymbol{A}}(E)
\label{eq::A Mittag Leffler non-degenerate}
\end{equation}
Again, under non-quasi-null eigenvectors assumption $ {\boldsymbol{b_{j}^{m}}}^\mathsf{T}\boldsymbol{b_{j}^{m}}\neq 0$, Gohberg-Sigal theory ensures the residues are normalized as:
\begin{equation}
\boldsymbol{a_j^m}{\boldsymbol{a_j^m}}^\mathsf{T}  = \frac{\boldsymbol{b_j^m}{\boldsymbol{b_j^m}}^\mathsf{T}}{{\boldsymbol{b_j^m}}^\mathsf{T} \left( \left. { \frac{\partial \boldsymbol{A}^{-1}}{\partial E} }\right|_{E=\mathcal{E}_j} \right) \boldsymbol{b_j^m}}
\label{eq:A residues }
\end{equation}
which is readily calculable from
\begin{equation}
\frac{\partial \boldsymbol{A}^{-1}}{\partial E} (\mathcal{E}_j) = - \Id{} - \boldsymbol{\gamma} \frac{\partial \boldsymbol{L} }{\partial E} (\mathcal{E}_j) \boldsymbol{\gamma}^\mathsf{T}
\label{eq: level matrix energy derivatives}
\end{equation}
Plugging (\ref{eq::A Mittag Leffler}) into (\ref{eq: def Kapur-Peierls operator}), and invoking the unicity of the complex residues, implies the radioactive widths (\ref{eq:R_L residues widths}) can be obtained as
\begin{equation}
\boldsymbol{r_j^m} = \boldsymbol{\gamma}^\mathsf{T}\boldsymbol{a_j^m}
\label{eq::radioactive widths rj from aj}
\end{equation}

This is an interesting and novel way to define the Siegert-Humblet parameters, which is similar to the alternative parameters definition of \cite{Ducru_shadow_Brune_Poles_2020}.
From this perspective, the alternative parameters appear as a special case that leave the Siegert-Humblet level-matrix parameters invariant to boundary condition $B_c$.
Indeed, one could search for the Siegert-Humblet expansion of the alternative parametrization of R-matrix theory, by simply proceeding as in equation (34) section III.A of \cite{Ducru_shadow_Brune_Poles_2020}, but replacing the level matrix $\boldsymbol{A}$ with the alternative level matrix $\boldsymbol{\widetilde{A}}$ (defined in equations (30) and (33) section III.A of \cite{Ducru_shadow_Brune_Poles_2020}):
\begin{equation}
\left.\boldsymbol{\widetilde{A}}^{-1}(E)\right|_{E=\mathcal{E}_j}\boldsymbol{\widetilde{b_j}} = \boldsymbol{0}
\label{eq::invA det roots on Brune}
\end{equation}
The exact same Gohberg-Sigal procedure can then be applied to the Mittag-Leffler expansion of the alternative level matrix $\boldsymbol{\widetilde{A}}$, in the vicinity $\mathcal{W}(E)$ of $E\in\mathbb{C}$ away from branch points $\left\{E_{T_c}\right\}$,
\begin{equation}
\boldsymbol{\widetilde{A}}(E) \underset{\mathcal{W}(E)}{=} \sum_{j\geq 1} \frac{\sum_{m=1}^{M_j}\boldsymbol{\widetilde{a_j^m}}{\boldsymbol{\widetilde{a_j^m}}}^\mathsf{T}}{E - \mathcal{E}_j} + \boldsymbol{\mathrm{Hol}}_{\boldsymbol{\widetilde{A}}}(E)
\label{eq:Brune A Mittag-Leffler}
\end{equation}
yielding the normalized residue widths:
\begin{equation}
\boldsymbol{\widetilde{a_j^m}} {\boldsymbol{\widetilde{a_j^m}}}^\mathsf{T}  = \frac{\boldsymbol{\widetilde{b_j^m}}{\boldsymbol{\widetilde{b_j^m}}}^\mathsf{T}}{{\boldsymbol{\widetilde{b_j^m}}}^\mathsf{T} \left( \left. { \frac{\partial \boldsymbol{\widetilde{A}}^{-1}}{\partial E} }\right|_{E=\mathcal{E}_j} \right) \boldsymbol{\widetilde{b_j^m}}}
\label{eq:A residues Brune}
\end{equation}
where (\ref{eq: level matrix energy derivatives}) can be combined to equation (33) of \cite{Ducru_shadow_Brune_Poles_2020}:
\begin{equation*}
\boldsymbol{\gamma^\mathsf{T} A \gamma} = \boldsymbol{\widetilde{\gamma}^\mathsf{T} \widetilde{A} \widetilde{\gamma}}
\label{eq:R_L unchanged by Brune}
\end{equation*}
to calculate the energy derivative.
Then, plugging (\ref{eq:A residues Brune}) into the same equation (33) of \cite{Ducru_shadow_Brune_Poles_2020}, we obtain the relation between the alternative R-matrix parameters and the Siegert-Humblet radioactive parameters:
\begin{equation}
\boldsymbol{r_j^m} = \boldsymbol{\widetilde{\gamma}}^\mathsf{T}\boldsymbol{\widetilde{a_j^m}}
\label{eq::radioactive widths rj from aj Brune}
\end{equation}

This relation (\ref{eq::radioactive widths rj from aj Brune}) is especially enlightening when compared to (\ref{eq::radioactive widths rj from aj}) from the viewpoint of invariance to boundary condition $B_c$. Indeed, we explained that the Siegert-Humblet parameters $\left\{\mathcal{E}_j, \boldsymbol{r_j^m}\right\}$ are invariant with a change of boundary condition $B_c \to B_c'$.
This is however not true of the level matrix residue widths $\left\{\boldsymbol{a_j^m}\right\}$ from (\ref{eq:A residues }). Thus, we can formally write this invariance by differentiating (\ref{eq::radioactive widths rj from aj}) with respect to $B_c$ and noting that $\frac{\partial \boldsymbol{r_j^m}}{\partial \boldsymbol{B}} = \boldsymbol{0}$, yielding:
\begin{equation}
\boldsymbol{0} = \frac{\partial \boldsymbol{\gamma}^\mathsf{T}}{\partial \boldsymbol{B}} \boldsymbol{a_j^m} + \boldsymbol{\gamma}^\mathsf{T}\frac{\partial \boldsymbol{a_j^m}}{\partial \boldsymbol{B}}
\label{eq:: invariance from Bc radioactive widths rj and aj}
\end{equation}
This new relation links the variation of the Wigner-Eisenbud resonance widths $\gamma_{\lambda,c}$ at level values $E_\lambda$ (resonance energies) under a change of boundary conditions $B_{c'}$, to the variation of the level matrix residue widths $a_{j,c}^m$ at pole values $\mathcal{E}_j$ under change of boundary condition $B_{c'}$. Since transformations (26) and (27) of section II.F in \cite{Ducru_shadow_Brune_Poles_2020} detail how to perform $\frac{\partial \boldsymbol{\gamma}^\mathsf{T}}{\partial \boldsymbol{B}}$, equation (\ref{eq:: invariance from Bc radioactive widths rj and aj}) could be used to update $\boldsymbol{a_j^m}$ under a change $B_c \to B_{c'}$.

Another telling insight from relation (\ref{eq:: invariance from Bc radioactive widths rj and aj}) is when we apply it to the relation between the alternative parameters and the Siegert-Humblet radioactive widths (\ref{eq::radioactive widths rj from aj Brune}). 
Since the alternative parameters $\boldsymbol{\widetilde{\gamma}}$ are invariant to $B_c$ (that is their main purpose), the same differentiation as in (\ref{eq:: invariance from Bc radioactive widths rj and aj}) now yields zero derivatives,
\begin{equation}
\boldsymbol{0} = \boldsymbol{\widetilde{\gamma}}^\mathsf{T}\frac{\partial \boldsymbol{\widetilde{a_j^m}}}{\partial \boldsymbol{B}}
\label{eq:: invariance from Bc radioactive widths rj and aj Brune}
\end{equation}
This is obvious from the fact that the alternative level matrix $\boldsymbol{\widetilde{A}}$ is invariant under change of boundary condition.
Yet invariance (\ref{eq:: invariance from Bc radioactive widths rj and aj Brune}) is insightful at it presents the alternative parameters $\left\{ \widetilde{E_i} , \widetilde{\gamma} \right\}$ as the ones which, when transformed to Siegert-Humblet radioactive parameters $\left\{\mathcal{E}_j , \boldsymbol{r_{j}} \right\}$ though (\ref{eq::radioactive widths rj from aj Brune}), leave the level residue widths $\left\{\boldsymbol{\widetilde{a_j}}\right\}$ invariant to $B_c$.

Conversely, the Kapur-Peierls pole expansion (\ref{eq::RL Mittag Leffler}) extends the alternative parametrization in that it generates boundary condition $B_c$ independent poles $\left\{\mathcal{E}_j\right\}$ and radioactive widths $\left\{\boldsymbol{r_j}\right\}$ that explicitly invert the alternative level matrix $\boldsymbol{\widetilde{A}}$ to yield (\ref{eq:Brune A Mittag-Leffler}).

\subsection{\label{sec: Siegert-Humblet Radioactive Pole Expansion Branch Structure} Siegert-Humblet Radioactive Pole Expansion Branch Structure}

Section \ref{sec:R_L def} introduced the Siegert-Humblet parametrization as the solutions of radioactive problem (\ref{eq:R_L radioactive problem}), where the $\boldsymbol{R}(E)$ matrix (\ref{eq:R expression}) is a function of the energy $E$, while $\boldsymbol{L}(\boldsymbol{\rho})$ is a function of the dimensionless wavenumber $\rho_c \triangleq a_c k_c(E)$. Thus, radioactive problem (\ref{eq:R_L radioactive problem}) can be solved either in energy space or in momentum space, both of which are linked by the $\rho(E)$ mapping (\ref{eq:rho_c(E) mapping}). 
This mapping induces a multi-sheeted Riemann surface, which introduces branch-points and sheets we now unveil in theorem \ref{theo::Siegert-Humblet Radioactive Poles}.

\begin{theorem}\label{theo::Siegert-Humblet Radioactive Poles} \textsc{Siegert-Humblet Radioactive Pole Expansion Branch Structure}. \\
Let the radioactive poles $\left\{\mathcal{E}_j\right\}$ be the solutions of the radioactive problem (\ref{eq:R_L radioactive problem}), and $\left\{E_{T_c}\right\}$ denote the threshold energies, branch-points of the $\rho_c(E)$ wavenumber-energy mapping (\ref{eq:rho_c(E) mapping}), then:
\begin{itemize}
    \item in the neighborhood $\mathcal{W}(E)$ of any complex energy $E$ away from branch-points $\left\{E_{T_c}\right\}$, there exists a series of complex matrices $\left\{ \boldsymbol{c}_n \right\}$ such that the Mittag-Leffler expansion (\ref{eq::RL Mittag Leffler}) takes the analytic form:
\begin{equation}
\boldsymbol{R}_{L}(E) \underset{\mathcal{W}(E)}{=} \sum_{j\geq 1} \frac{\boldsymbol{r_j}\boldsymbol{r_j}^\mathsf{T}}{E - \mathcal{E}_j} +\sum_{n\geq0} \boldsymbol{c}_n E^n
\label{eq::RL Mittag Leffler expanded}
\end{equation}
    \item the radioactive poles $\left\{\mathcal{E}_j\right\}$ are complex, and live on the multi-sheeted Riemann surface of $k_c(E)$ wavenumber-energy mapping (\ref{eq:rho_c(E) mapping}):
\begin{equation}
\Big\{\mathcal{E}_j, + , +, -, \hdots, +, - \Big\}
\label{eq:: pole E_j sheet reporting}
\end{equation}
    \item Let $N_L$ be the number of solutions to the radioactive problem (\ref{eq:R_L radioactive problem}) in wavenumber $\boldsymbol{\rho}$ space. For neutral particles:
\begin{equation}
N_L = \left(2 N_\lambda  + \sum_{c=1}^{N_c} \ell_c \right)\; \times 2^{(N_{E_{T_c} \neq E_{T_{c'}}} - 1 ) }
\label{eq::NL number of poles}
\end{equation}
    \item For charged particles, there is an infinite number (countable) of radioactive poles: $N_L = \infty$.
\end{itemize}
\end{theorem}

\begin{proof}
Away from the branch points $\left\{E_{T_c}\right\}$, the holomorphic part of Mittag-Leffler expansion (\ref{eq::RL Mittag Leffler}) can be analytically expanded in series as (\ref{eq::RL Mittag Leffler expanded}) -- we here assumed the non-degenerate case of rank-one residues (multiplicity $M_j = 1$) though it is readily generalizable to (\ref{eq::RL Mittag Leffler degenerate state}). 

When solving radioactive problem (\ref{eq:R_L radioactive problem}), or (\ref{eq: invA eigenproblem for RL poles}), to obtain the Siegert-Humblet poles $\left\{\mathcal{E}_j\right\}$ and residues $\left\{\boldsymbol{r_j}\right\}$, or $\left\{\boldsymbol{a_j}\right\}$, it is necessary to compute the $\boldsymbol{L^0}$ matrix function $\boldsymbol{L^0}(E) \triangleq  \boldsymbol{L^0}(\boldsymbol{\rho}(E)) $ for complex energies $E\in \mathbb{C}$. As discussed in \cite{Ducru_shadow_Brune_Poles_2020} (c.f. sections II.A, B, III. B. and C. of \cite{Ducru_shadow_Brune_Poles_2020}), mapping (\ref{eq:rho_c(E) mapping}) generates a multi-sheeted Riemann surface with $2^{N_c}$ branches (with the threshold values $E_{T_c}$ as branch points), corresponding to the choice for each channel $c$, of the sign of the square root in $\boldsymbol{\rho}(E)$. This means that when searching for the poles, one has to keep track of these choices and specify for each pole $\mathcal{E}_j$ on what sheet it is found. Every pole $\mathcal{E}_j$ must thus come with the full reporting of these $N_c$ signs, i.e. $\left\{\mathcal{E}_j, - , +, +, \hdots, -, + \right\} $ as (\ref{eq:: pole E_j sheet reporting}).

 When searching for these radioactive poles in wavenumber space, the $\boldsymbol{R}_{L}$ Kapur-Peierls operator (\ref{eq: def Kapur-Peierls operator}) is continued to complex wavenumbers by meromorphic continuation of $\boldsymbol{L}(\boldsymbol{\rho})$, where the reduced logarithmic derivative of the outgoing wavefunction (\ref{eq: L operator}) takes the Mittag-Leffler expansion described in equation (13) of theorem 1 section II.B. of \cite{Ducru_shadow_Brune_Poles_2020}. 
There are more radioactive poles $\left\{\mathcal{E}_j\right\}$ than Wigner-Eisenbud levels $\left\{E_\lambda\right\}$ --- as was the case for the alternative parameters (c.f. theorems 2 and 3 III.C of \cite{Ducru_shadow_Brune_Poles_2020}).
For massive neutral particles, we can proceed in an analogous fashion as for the proof of theorem 3 section III.C of \cite{Ducru_shadow_Brune_Poles_2020}, and apply the diagonal divisibility and capped multiplicities lemma (lemma 3 section III.C. of \cite{Ducru_shadow_Brune_Poles_2020}) to the determinant of the Kapur-Peierls operator $\boldsymbol{R}_L$ in (\ref{eq:R_L radioactive problem}) -- but this time in $\rho_c$ space -- and then look at the order of the resulting rational fractions in $\rho_c$ and the number of times one must square the polynomials to unfold all $\rho_c = \mp \sqrt{\cdot}$ sheets of mapping (\ref{eq:rho_c(E) mapping}). We were thus able to establish that the number $N_L$ of poles in wavenumber $\rho$-space is $2 N_\lambda  + \sum_{c=1}^{N_c} \ell_c$ poles per sheet: that is a total of (\ref{eq::NL number of poles}) over all sheets, where $N_{E_{T_c} \neq E_{T_{c'}}}$ designates the number of different thresholds (including the obvious $E_{T_c} = 0$ zero threshold) and thus the number of sheets.

In the charged particles case, $L_c(\rho_c)$ has a countably infinite number of poles, which in turn induces an infinite number (countable) of solutions to the radioactive problem (\ref{eq:R_L radioactive problem}) (c.f. discussion after theorem 1 section II.B of \cite{Ducru_shadow_Brune_Poles_2020}).

\end{proof}

It is important to grasp the meaning of the Mittag-Leffler expansion (\ref{eq::RL Mittag Leffler}) --- or (\ref{eq::A Mittag Leffler}) and (\ref{eq:Brune A Mittag-Leffler}). These are local expressions: they do not hold for all complex energies $E\in \mathbb{C}$ because of the branch-point structure of the Riemann sheet. However, in the neighborhood $\mathcal{W}(E)$ of any complex energy point $E\in\mathbb{C}$ away from the branch-points (thresholds $\left\{E_{T_c}\right\}$), the Mittag-Leffler expansion (\ref{eq::RL Mittag Leffler}) is true, and its holomorphic part admits an analytic expansion $\boldsymbol{\mathrm{Hol}}_{\boldsymbol{R}_{L}}(E) \triangleq \sum_{n\geq0} \boldsymbol{c}_n E^n$.
This has two major consequences for the Siegert-Humblet expansion.
First, contrarily to the alternative parameters $\left\{\widetilde{E_i}, \widetilde{\gamma_{i,c}}\right\}$ discussed in \cite{Ducru_shadow_Brune_Poles_2020}, the Siegert-Humblet set of radioactive poles and widths $\left\{\mathcal{E}_j , r_{j,c} \right\}$ do not suffice to uniquely determine the energy behavior of the scattering matrix $\boldsymbol{U}(E)$: one needs to locally add the expansion coefficients $\left[\boldsymbol{c}_n\right]_{c,c'}$ of the entire part $\boldsymbol{\mathrm{Hol}}_{\boldsymbol{R}_{L}}(E) \triangleq \sum_{n\geq0} \boldsymbol{c}_n E^n$.
Second, since the set of coefficients $\left\{\boldsymbol{c}_n\right\}$ is \textit{a priori} infinite (and so is the set of poles in the Coulomb case), this means that numerically the Siegert-Humblet expansion can only be used to compute local approximations of the scattering matrix, which can nonetheless reach any target accuracy by increasing the number of $\left\{  \mathcal{E}_j\right\}_{j \in \llbracket 1,N_L \rrbracket} $ poles included and the order of the truncation $N_{\mathcal{W}(E)}$ in $\left\{\boldsymbol{c}_n\right\}_{n\in\llbracket1,N_{\mathcal{W}(E)}\rrbracket}$.
In practice, this means that to compute the scattering matrix one needs to provide the Siegert-Humblet parameters $\left\{\mathcal{E}_j, r_{j,c}\right\}$, cut the energy domain of interest into local windows $\mathcal{W}(E)$ away from threshold branch-points $\left\{E_{T_c}\right\}$, and provide a set of local coefficients $\left\{\boldsymbol{c}_n\right\}_{n\in\llbracket1,N_{\mathcal{W}(E)}\rrbracket}$ for each window.

As discussed in \cite{Ducru_shadow_Brune_Poles_2020}  (c.f. lemmas 1 and 2 section III.B and theorems 2 and 3 of section III.C of \cite{Ducru_shadow_Brune_Poles_2020}), the definition of the shift and penetration functions for complex wavenumbers is ambiguous (in particular purely imaginary wavenumbers yield negative or sub-threshold energies), which in turn entail various possible alternative parameters. 
When solving radioactive problem (\ref{eq:R_L radioactive problem}) to find the Siegert-Humblet radioactive poles and residues $\left\{ \mathcal{E}_j, r_{j,c}\right\}$ -- or (\ref{eq: invA eigenproblem for RL poles}) equivalently -- there are no such ambiguities on the definition of $\boldsymbol{L}$: the Kapur-Peierls operator is simply analytically continued to complex wavenumbers.
The unicity of analytic continuation thus entails that the Siegert-Humblet parameters are uniquely defined, as long as we specify for each channel $c$ what sheet of the Riemann surface from mapping (\ref{eq:rho_c(E) mapping}) was chosen, as in (\ref{eq:: pole E_j sheet reporting}).

The $\left\{\mathcal{E}_j, + , +, \hdots, +, + \right\}$ sheet is called the \textit{physical sheet}, and we here call the poles on that sheet the \textit{principal poles}. All other sheets are called \textit{nonphysical} and the poles laying on these sheets are called \textit{shadow poles}. Often, the principal poles are responsible for the resonant behavior, with shadow poles only contributing to background behavior, but cases have emerged where the shadow poles contribute significantly to the resonance structure, as reported in \cite{Hale_1987}, and G. Hale there introduced a quantity called \textit{strength} of a pole (c.f. eq. (7) in \cite{Hale_1987}, or paragraph after eq. (2.11) XI.2.b, p.306, and section XI.4, p.312 in \cite{Lane_and_Thomas_1958}) to quantify the impact a pole $\mathcal{E}_j$ will have on resonance behavior, by comparing the residue $r_{j,c}$ to the Wigner-Eisenbud widths $\gamma_{\lambda,c}$.\\

Result (\ref{eq::NL number of poles}) is quite instructive: one can observe that the number $N_L$ of Siegert-Humblet poles adds-up the number of levels $N_\lambda$ and the number of poles of $\boldsymbol{L}$ (which is $\sum_{c=1}^{N_c} \ell_c$ for neutral massive particles, and is infinite in the Coulomb case, c.f. discussion after theorem 1 in section II.B of \cite{Ducru_shadow_Brune_Poles_2020}. Moreover, $N_L$ is duplicated with each new sheet of the Riemann surface from mapping (\ref{eq:rho_c(E) mapping}) --- that is associated to a new threshold, hence the $N_{E_{T_c} \neq E_{T_{c'}}}$.
Interestingly, comparing $N_L$ from (\ref{eq::NL number of poles}) with the number $N_S$ of alternative analytic poles from equation (49) in theorem 3 section III.C of \cite{Ducru_shadow_Brune_Poles_2020} --- which are in $E$-space and must thus be doubled to obtain the number of $\rho$-space poles --- we note that the analytic continuation of the shift factor $\boldsymbol{S}$ (c.f. lemma 2 III.B of \cite{Ducru_shadow_Brune_Poles_2020}) adds a virtual pole for each pole of $\boldsymbol{L}$ when unfolding the sheets of mapping (\ref{eq:rho_c(E) mapping}), because it is a function of $\rho_c^2(E)$. This can readily be observed in the trivial case of a p-wave ($\ell=1$) channel with one resonance (one level $N_\lambda=1$), where $S(E) = -\frac{1}{1 + \rho^2(E)}$ introduces two poles at $\rho(E) = \pm \mathrm{i}$, while $L(E) = \frac{- 1 + \mathrm{i}\rho(E) + \rho^2(E)}{1 - \mathrm{i}\rho(E)}$ only counts one pole, at $\rho(E) = \mathrm{i}$.

As for equation (50) of theorem 3 section III.C. of \cite{Ducru_shadow_Brune_Poles_2020}, one should add the precision that in the sum over the channels in (\ref{eq::NL number of poles}), the multiplicity of possible $L_c(\rho_c)$ repeated over many different channels $L_c(\rho_c) = L_{c'\neq c}(\rho_{c'})$ is capped by $N_\lambda$, which in practice would only occur in the rare cases where only one or two levels occurs for many channels with same angular momenta (and, of course, total angular momenta and parity $J^\pi$). 

Numerically, solving the generalized eigenvalue problems (\ref{eq:R_L radioactive problem}) or (\ref{eq: invA eigenproblem for RL poles}) falls into the well-known class of nonlinear eigenvalue problems, for which algorithms we direct the reader to Heinrich Voss's chapter 115 in the Handbook of Linear Algebra \cite{Handbook_of_linear_algebra}.
We will just state that instead of the Rayleigh-quotient type of methods expressed in \cite{Handbook_of_linear_algebra}, it can sometimes be computationally advantageous to first find the poles $\left\{\mathcal{E}_j\right\}$ by solving the channel determinant problem, $\mathrm{det}\left( \left.\boldsymbol{R}_{L}^{-1}(E)\right|_{E = \mathcal{E}_j} \right) = 0$, or the corresponding level determinant one, $\mathrm{det}\left( \left.\boldsymbol{A}^{-1}(E)\right|_{E = \mathcal{E}_j} \right) = 0$, and then solve the associated linear eigenvalue problem. Methods tailored to find all the roots of this problem where introduced in \cite{Ducru_PHYSOR_conversion_2016}, or in equations (200) and (204) of \cite{Frohner_Jeff_2000}.
Notwithstanding, from a numerical standpoint, having the two approaches is beneficial in that solving (\ref{eq:R_L radioactive problem}) will be advantageous over solving (\ref{eq: invA eigenproblem for RL poles}) when the number of levels $N_\lambda$ far exceeds the number of channels $N_c$, and conversely.
Nonetheless, the multi-sheeted nature of the radioactive problem makes it harder to solve, as one must search each sheet of mapping (\ref{eq:rho_c(E) mapping}) to find all the poles.

\subsection{\label{sec: Xenon evidence of radioactive parameters}Xenon $^{\mathrm{134}}\mathrm{Xe}$ evidence of radioactive parameters}

\begin{figure}[] 
  \centering
  \subfigure[\ First p-wave resonance.]{\includegraphics[width=0.49\textwidth]{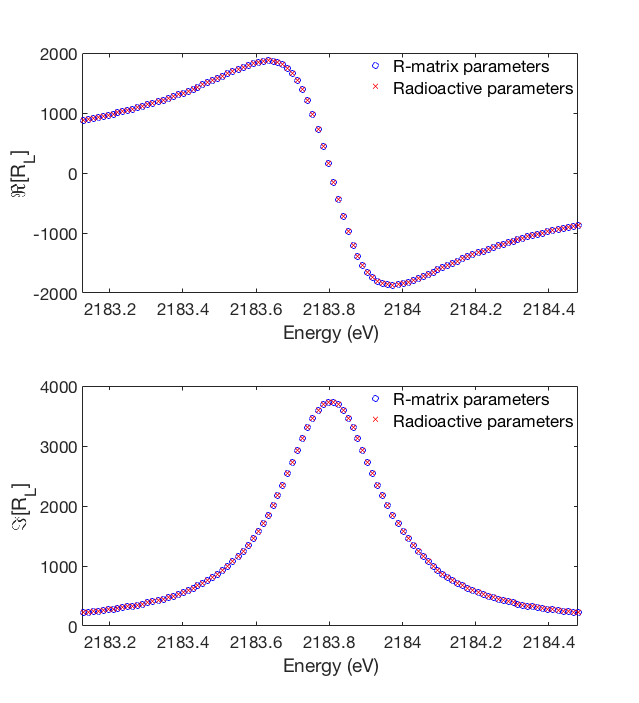}}
  \subfigure[\ Second p-wave resonance.]{\includegraphics[width=0.49\textwidth]{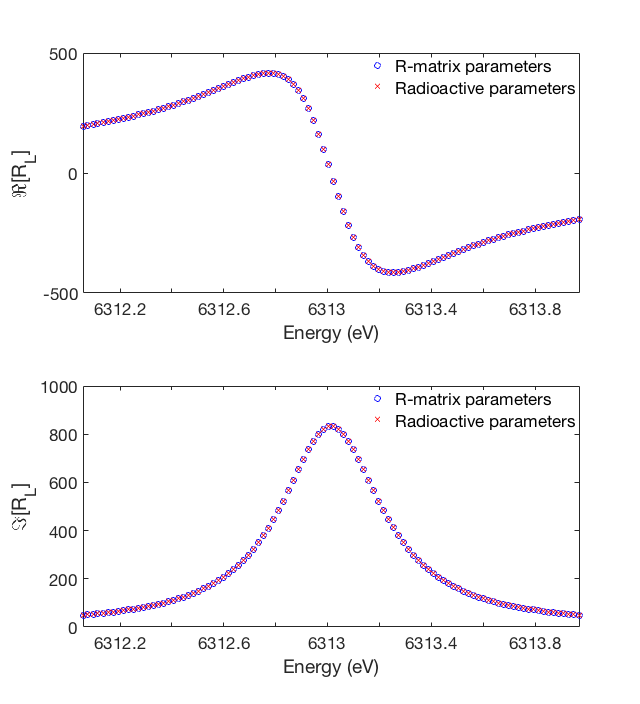}}
  \caption{\small{Kapur-Peirels $\boldsymbol{R}_L(E)$ operator (\ref{eq: def Kapur-Peierls operator}) of xenon $^{\mathrm{134}}\mathrm{Xe}$ two p-wave resonances in spin-parity group $J^\pi = 1/2^{(-)}$. Dimensionless $\boldsymbol{R}_L(E)$ is computed using radioactive parameters from table \ref{tab:Xe-134 radioactive parameters} in expression (\ref{eq::xenon RL Mittag Leffler expanded in E space}), or using the R-matrix parameters from table \ref{tab:Xe-134 radioactive parameters} in Reich-Moore level-matrix (\ref{eq:inv_A expression}) -- that is definition (21) of \cite{Ducru_WMP_Phys_Rev_C_2019} -- yielding identical real and imaginary parts.
  }}
  \label{fig:xenon-134 J=1/2(-) Kapur-Peierls operator}
\end{figure}

\begin{table}[ht!!]
\caption{\label{tab:Xe-134 radioactive parameters} 
Radioactive parameters (Siegert-Humblet poles and residue widths of the Kapur-Peierls $\boldsymbol{R}_L(E)$ operator) of the two p-wave resonances of $^{\mathrm{134}}\mathrm{Xe}$, spin-parity group $J^\pi = 1/2^{(-)}$, converted from ENDF/B-VIII.0 evaluation (MLBW) to multipole representation using Reich-Moore level matrix (\ref{eq:inv_A expression}), that is definition (21) of \cite{Ducru_WMP_Phys_Rev_C_2019}.}
\begin{ruledtabular}
\begin{tabular}{l}
$ z = \sqrt{E}$ with $E$ in (eV)  \tabularnewline
$A = 132.7600$ \tabularnewline
$a_c = 5.80$ : channel radius (Fermis) \tabularnewline
$\rho_0 = \frac{A a_c \sqrt{\frac{2m_n}{h}} }{A + 1} $ in ($\mathrm{\sqrt{eV}}^{-1}$), so that $\rho(z) \triangleq \rho_0 z$ \tabularnewline
with $\sqrt{\frac{2m_n}{h}} = 0.002196807122623 $ in units ($1/(10^{-14}\text{m} \mathrm{\sqrt{eV}})$) \tabularnewline
\hline \tabularnewline
\textbf{Radioactive parameters} (rounded to 5 digits): \tabularnewline
\hline 
\begin{tabular}{l|l|l}
 Radioactive poles  & Radioactive residue & Level-matrix residue  \tabularnewline
   $\left\{\mathcal{E}_j , \pm \right\}$ from (\ref{eq: invA eigenproblem for RL poles}) & widths $\boldsymbol{r_j}$ from (\ref{eq::radioactive widths rj from aj}) & widths $\boldsymbol{a_j}$ from (\ref{eq::A Mittag Leffler non-degenerate}) \tabularnewline
  ($\mathrm{eV}$), sheet of (\ref{eq:rho_c(E) mapping}) & ($\mathrm{\sqrt{eV}}$)  &  (dimensionless)  \tabularnewline
\hline \hline
$\mkern-5mu\left\{\mkern-10mu\begin{array}{c}
   -6.2694 {\mkern+2mu\mathrm{E}\mkern-1mu {+} \mkern-1mu {5} }     \\
    -\mathrm{i} 1.0238 {\mkern+2mu\mathrm{E} \mkern-1mu {-} \mkern-1mu {4} }
\end{array} \mkern-5mu, \mkern-5mu + \mkern-5mu \right\} \mkern-5mu$
& $\begin{array}{c}
   9.1193 {\mkern+2mu\mathrm{E}\mkern-1mu {-} \mkern-1mu {8} }     \\
    -\mathrm{i} 1.4762 {\mkern+2mu\mathrm{E} \mkern-1mu {+} \mkern-1mu {0} }
\end{array} $
& $\left[\begin{array}{c}
   2.7683{\mkern+2mu\mathrm{E}\mkern-1mu {-} \mkern-1mu {9} }     \\
    -\mathrm{i} 4.4744{\mkern+2mu\mathrm{E} \mkern-1mu {-} \mkern-1mu {2} }  \\ \hline
    1.5345 {\mkern+2mu\mathrm{E}\mkern-1mu {-} \mkern-1mu {9} }     \\
    -\mathrm{i} 2.4964{\mkern+2mu\mathrm{E} \mkern-1mu {-} \mkern-1mu {2} }
\end{array} \right]$
\tabularnewline \hline
$\mkern-5mu\left\{\mkern-10mu\begin{array}{c}
   2.1838 {\mkern+2mu\mathrm{E}\mkern-1mu {+} \mkern-1mu {3} }     \\
    +\mathrm{i} 9.0757 {\mkern+2mu\mathrm{E} \mkern-1mu {-} \mkern-1mu {2} }
\end{array} \mkern-5mu, \mkern-5mu - \mkern-5mu \right\} \mkern-5mu$
& $\begin{array}{c}
   8.6799 {\mkern+2mu\mathrm{E}\mkern-1mu {-} \mkern-1mu {4} }     \\
    -\mathrm{i} 2.5113  {\mkern+2mu\mathrm{E} \mkern-1mu {+} \mkern-1mu {1} }
\end{array} $
& $\left[\begin{array}{c}
   4.444{\mkern+2mu\mathrm{E}\mkern-1mu {-} \mkern-1mu {5} }     \\
    -\mathrm{i} 9.995{\mkern+2mu\mathrm{E} \mkern-1mu {-} \mkern-1mu {1} }  \\ \hline
    -1.7608 {\mkern+2mu\mathrm{E}\mkern-1mu {-} \mkern-1mu {5} }     \\
    -\mathrm{i} 2.9849 {\mkern+2mu\mathrm{E} \mkern-1mu {-} \mkern-1mu {4} }
\end{array} \right]$
\tabularnewline \hline
$\mkern-5mu\left\{\mkern-10mu\begin{array}{c}
   2.1838 {\mkern+2mu\mathrm{E}\mkern-1mu {+} \mkern-1mu {3} }     \\
    -\mathrm{i}  1.6868 {\mkern+2mu\mathrm{E} \mkern-1mu {-} \mkern-1mu {1} }
\end{array}  \mkern-5mu, \mkern-5mu + \mkern-5mu \right\} \mkern-5mu$
& $\begin{array}{c}
   8.6814  {\mkern+2mu\mathrm{E}\mkern-1mu {-} \mkern-1mu {4} }     \\
    +\mathrm{i} 2.5113 {\mkern+2mu\mathrm{E} \mkern-1mu {+} \mkern-1mu {1} }
\end{array} $
& $\left[\begin{array}{c}
   4.444{\mkern+2mu\mathrm{E}\mkern-1mu {-} \mkern-1mu {5} }     \\
    +\mathrm{i} 9.995{\mkern+2mu\mathrm{E} \mkern-1mu {-} \mkern-1mu {1} }  \\ \hline
    -1.7597 {\mkern+2mu\mathrm{E}\mkern-1mu {-} \mkern-1mu {5} }     \\
    +\mathrm{i} 2.9849 {\mkern+2mu\mathrm{E} \mkern-1mu {-} \mkern-1mu {4} }
\end{array} \right]$
\tabularnewline \hline
$\mkern-5mu\left\{\mkern-10mu\begin{array}{c}
   6.3130 {\mkern+2mu\mathrm{E}\mkern-1mu {+} \mkern-1mu {3} }     \\
    +\mathrm{i} 1.6025 {\mkern+2mu\mathrm{E} \mkern-1mu {-} \mkern-1mu {1} }
\end{array}  \mkern-5mu, \mkern-5mu - \mkern-5mu \right\} \mkern-5mu$
& $\begin{array}{c}
    2.4919 {\mkern+2mu\mathrm{E}\mkern-1mu {-} \mkern-1mu {3} }     \\
    -\mathrm{i} 1.4085  {\mkern+2mu\mathrm{E} \mkern-1mu {+} \mkern-1mu {1} }
\end{array} $
& $\left[\begin{array}{c}
   8.5974 {\mkern+2mu\mathrm{E}\mkern-1mu {-} \mkern-1mu {5} }     \\
    +\mathrm{i} 8.5696 {\mkern+2mu\mathrm{E} \mkern-1mu {-} \mkern-1mu {4} }  \\ \hline
    2.3534 {\mkern+2mu\mathrm{E}\mkern-1mu {-} \mkern-1mu {5} }     \\
    -\mathrm{i} 9.9984 {\mkern+2mu\mathrm{E} \mkern-1mu {-} \mkern-1mu {1} }
\end{array} \right]$
\tabularnewline \hline
$\mkern-5mu\left\{\mkern-10mu\begin{array}{c}
  6.3130 {\mkern+2mu\mathrm{E}\mkern-1mu {+} \mkern-1mu {3} }     \\
    -\mathrm{i} 2.3822 {\mkern+2mu\mathrm{E} \mkern-1mu {-} \mkern-1mu {1} }
\end{array}  \mkern-5mu, \mkern-5mu + \mkern-5mu \right\} \mkern-5mu$
& $\begin{array}{c}
   2.4916 {\mkern+2mu\mathrm{E}\mkern-1mu {-} \mkern-1mu {3} }     \\
    + \mathrm{i} 1.4085 {\mkern+2mu\mathrm{E} \mkern-1mu {+} \mkern-1mu {1} }
\end{array} $
& $\left[\begin{array}{c}
   8.5964 {\mkern+2mu\mathrm{E}\mkern-1mu {-} \mkern-1mu {5} }     \\
    -\mathrm{i} 8.5697 {\mkern+2mu\mathrm{E} \mkern-1mu {-} \mkern-1mu {4} }  \\ \hline
   2.3534 {\mkern+2mu\mathrm{E}\mkern-1mu {-} \mkern-1mu {5} }     \\
    +\mathrm{i} 9.9984 {\mkern+2mu\mathrm{E} \mkern-1mu {-} \mkern-1mu {1} }
\end{array} \right]$
\end{tabular} \\
\hline\\
\textbf{R-matrix parameters}: \\
\hline \\
$E_1 = 2186.0$ : first resonance energy (eV) \\
$\Gamma_{1,n} = 0.2600$ : neutron width of first resonance \\ (not reduced width), i.e. $\Gamma_{\lambda,c} = 2P_c(E_\lambda) \gamma_{\lambda,c}^2$ \\
$\Gamma_{1,\gamma} = 0.0780$ : eliminated capture width (eV) \\
$E_2 = 6315.0$ : second resonance energy (eV) \\
$\Gamma_{2,n} = 0.4000$ (eV) \\
$\Gamma_{2,\gamma} = 0.0780$ (eV) \\
$g_{J^\pi} = 1/3$ : spin statistical factor \\
$B_c = - 1 $ \\
\end{tabular}
\end{ruledtabular}
\end{table}
















%


%



In our previous article \cite{Ducru_shadow_Brune_Poles_2020}, we observed the first evidence of shadow poles in the alternative parametrization of R-matrix theory in isotope xenon $^{\mathrm{134}}\mathrm{Xe}$ spin-parity group $J^\pi = 1/2^{(-)}$, showing how they depend on the choice of continuation to complex wavenumbers. 
We here document in table \ref{tab:Xe-134 radioactive parameters} the Siegert-Humblet radioactive parameters (poles and residues of the Kapur-Peierls $\boldsymbol{R}_L$ operator), for these same p-wave resonances of $^{\mathrm{134}}\mathrm{Xe}$ spin-parity group $J^\pi = 1/2^{(-)}$. As shown in figure \ref{fig:xenon-134 J=1/2(-) Kapur-Peierls operator}, both the radioactive parameters and the R-matrix parameters yield an identical Kapur-Peierls $\boldsymbol{R}_L(E)$ operator, and therefore exactly reconstruct the scattering matrix $\boldsymbol{U}(E)$ of the nuclear interactions. 

This is made possible because $^{\mathrm{134}}\mathrm{Xe}$ spin-parity group $J^\pi = 1/2^{(-)}$ has only a neutron channel with zero threshold ($E_{T_c} = 0$).
In the particular case of neutral particles with zero threshold, the outgoing wavefunction reduced logarithmic derivative operator $\boldsymbol{L}(\boldsymbol{\rho})$ is a rational function in $\sqrt{E}$: this can be seen from Mittag-Leffer expansion (\ref{eq::Explicit Mittag-Leffler expansion of L_c}) with a finite amount of poles $\left\{\omega_n \right\}$ (reported in theorem 1 section II.B and table II of \cite{Ducru_shadow_Brune_Poles_2020}).
Therefore, the transformation $z \triangleq \sqrt{E}$ unfolds the Riemann surface of mapping (\ref{eq:rho_c(E) mapping}): that is that searching for solutions to the radioactive problem (\ref{eq:R_L radioactive problem}) in $z$-space is equivalent to searching on both sheets of the $\pm\sqrt{E}$ Riemann surface from mapping (\ref{eq:rho_c(E) mapping}).
Moreover, a study of the numerator and denominator of the inverse level matrix $\boldsymbol{A^{-1}}(z)$ from (\ref{eq:inv_A expression}) then shows that the level matrix $\boldsymbol{A}(z)$ is rational function of degree -2 in $z$-space, with $N_L$ poles from (\ref{eq::NL number of poles}) with only one sheet (no other thresholds than zero), so that its Mittag-Leffler expansion (\ref{eq::A Mittag Leffler non-degenerate}) is actually a partial fraction decomposition in simple $z$ poles, without constant nor holomorphic part (c.f. section II.F. of \cite{Ducru_WMP_Phys_Rev_C_2019} for more in depth discussion of this process):
\begin{equation}
\boldsymbol{A}(E) = \sum_{j = 1}^{N_L} \frac{\frac{\boldsymbol{a_j}\boldsymbol{a_j}^\mathsf{T}}{2\sqrt{\mathcal{E}_j}}}{\sqrt{E} - \sqrt{\mathcal{E}_j}}
\label{eq::xenon A Mittag Leffler expanded in sqrt(E) space}
\end{equation}
Note that the nullity of the constant term entails the following remarkable property:
\begin{equation}
 \sum_{j = 1}^{N_L} \frac{\boldsymbol{a_j}\boldsymbol{a_j}^\mathsf{T}}{2\sqrt{\mathcal{E}_j}} = \boldsymbol{0}
\label{eq::xenon A residues add to nullity}
\end{equation}
Since from (\ref{eq::radioactive widths rj from aj}) we have $\boldsymbol{r_j} = \boldsymbol{\gamma}^\mathsf{T}\boldsymbol{a_j}$ (assuming non-degenerate states), the latter properties on the level-matrix can be transcribed into the following exact pole expansion for the Kapur-Peierls operator (\ref{eq: def Kapur-Peierls operator}):
\begin{equation}
\boldsymbol{R}_{L}(E) = \sum_{j = 1}^{N_L} \frac{\frac{\boldsymbol{r_j}\boldsymbol{r_j}^\mathsf{T}}{2\sqrt{\mathcal{E}_j}}}{\sqrt{E} - \sqrt{\mathcal{E}_j}}
\label{eq::xenon RL Mittag Leffler expanded in sqrt(E) space}
\end{equation}
which is equivalent to eq. (106) of \cite{Ducru_WMP_Phys_Rev_C_2019}, and the null constant relations (\ref{eq::xenon A residues add to nullity}) entails the remarkable property on the radioactive parameters (c.f. eq. (108) of \cite{Ducru_WMP_Phys_Rev_C_2019}):
\begin{equation}
 \sum_{j = 1}^{N_L} \frac{\boldsymbol{r_j}\boldsymbol{r_j}^\mathsf{T}}{2\sqrt{\mathcal{E}_j}} = \boldsymbol{0}
\label{eq::xenon RL residues add to nullity}
\end{equation}
By setting a choice of sheet in $z = \pm \sqrt{E}$, the latter equations can be written as:
\begin{equation}
\boldsymbol{R}_{L}(E) = \sum_{j = 1}^{N_L} \frac{\boldsymbol{r_j}\boldsymbol{r_j}^\mathsf{T}}{E - \mathcal{E}_j} + \underbrace{\sum_{j = 1}^{N_L} \frac{\frac{\boldsymbol{r_j}\boldsymbol{r_j}^\mathsf{T}}{2\sqrt{\mathcal{E}_j}}}{\sqrt{E} + \sqrt{\mathcal{E}_j}}}_{\boldsymbol{\mathrm{Hol}}_{\boldsymbol{R}_{L}}(E)}
\label{eq::xenon RL Mittag Leffler expanded in E space}
\end{equation}
where $-\sqrt{\mathcal{E}_j}$ is not a pole, and therefore the second term is the exact holomorphic part $\boldsymbol{\mathrm{Hol}}_{\boldsymbol{R}_{L}}(E)$ from (\ref{eq::RL Mittag Leffler}).

\subsection{\label{sec: Wenon evidence of radioactive parameters}Comparing radioactive, traditional, and alternative R-matrix parameters}

This case of xenon $^{\mathrm{134}}\mathrm{Xe}$ shows the general merits of the radioactive parameters: in contrast
with the R-matrix resonance parameters, the radioactive poles $\mathcal{E}_j$ are independent of both the arbitrary boundary parameter $B_c$ and the channel radius $a_c$, while the radioactive widths $\boldsymbol{r_j}$ are independent of the boundary parameters $B_c$ and depend on the channel radius in a systematic
way (provided by theorem \ref{theo:: Radioactive parameters transformation under change of channel radius} bellow).
Moreover, in this specific neutral particles with zero-threshold case, the Kapur-Peierls $\boldsymbol{R}_L(E)$ operator pole expansions (\ref{eq::xenon RL Mittag Leffler expanded in sqrt(E) space}) and (\ref{eq::xenon RL Mittag Leffler expanded in E space}) are exact, and therefore can fully reconstruct the R-matrix model scattering matrix, as shown by figure \ref{fig:xenon-134 J=1/2(-) Kapur-Peierls operator}.

Nonetheless, this example also shows the limitations of the radioactive parameters pole expansion approach (\ref{eq::RL Mittag Leffler expanded}) of theorem \ref{theo::Siegert-Humblet Radioactive Poles}.
Just as the alternative parameters of Brune in \cite{Ducru_shadow_Brune_Poles_2020}, the radioactive parameters entangle the energy dimension with the wavenumber one, meaning one now has to specify with each radioactive pole $\mathcal{E}_j$ its sheet (\ref{eq:: pole E_j sheet reporting}) on the Riemann surface of mapping (\ref{eq:rho_c(E) mapping}), for each threshold branch, as specified in theorem \ref{theo::Siegert-Humblet Radioactive Poles}. 
In contrast, though they depend on the arbitrary boundary parameters $B_c$ and channel radii $a_c$, the traditional Wigner-Eisenbud R-matrix parameters have the truly remarkable (and seldom appreciated) property of de-entangling the energy dimension from the wavenumber one. The Wigner-Eisenbud resonance parameters are real and well defined in energy space, without any need to map to the wavenumber and therefore specify where the resonance energies $E_\lambda$ dwell on the multi-sheeted Riemann surface of mapping (\ref{eq:rho_c(E) mapping}).

Another significant limitation of the radioactive parameters is that they are in general incomplete, meaning that the knowledge of the radioactive poles and residues is not sufficient to fully parametrize the $\boldsymbol{R}_L(E)$ Kapur-Peierls operator: one also needs to parametrize the holomorphic part $\boldsymbol{\mathrm{Hol}}_{\boldsymbol{R}_{L}}(E)$ in Mittag-Leffler expansion (\ref{eq::RL Mittag Leffler expanded}). In the general case of charged particles or thresholds, there is no simple way of parametrizing this holomorphic part (though it is known exactly for zero-threshold neutral particles as equation (\ref{eq::xenon RL Mittag Leffler expanded in E space}) specifies). Moreover, even if the holomorphic part were known, in the general case of charged particles and thresholds there is an infinite number of radioactive poles ($N_L = \infty$), all of which are necessary to exactly reconstruct the scattering matrix. 
This means the radioactive parameters alone are not very well suited for evaluations in standard nuclear data libraries. Nonetheless, the radioactive poles have recently been used to constitute an alternative nuclear data library --- the \textit{windowed multipole library} --- with the goal of achieving significant computational performance gains in nuclear simulations, as we explain in our follow-up article \cite{Ducru_WMP_Phys_Rev_C_2019}: the final in the xenon trilogy on pole parametrizations of R-matrix theory \cite{Ducru_shadow_Brune_Poles_2020, Ducru_Scattering_Matrix_of_Complex_Wavenumbers_2019, Ducru_WMP_Phys_Rev_C_2019}.

For comparison, the alternative parameters proposed by Brune in \cite{Brune_2002} combine some merits and drawbacks of both the radioactive and the traditional (Wigner-Eisenbud) parameters.
Like the radioactive parameters, the alternative parameters are independent of the arbitrary boundary condition $B_c$, though they still depend on arbitrary channel radii $a_c$.
Like the Wigner-Eisenbud resonance parameters, the alternative parameters are always complete: with the knowledge of $N_\lambda$ alternative poles, one can fully reconstruct the scattering matrix (c.f. theorem 4 of \cite{Ducru_shadow_Brune_Poles_2020}).
On the other hand, unlike the Wigner-Eisenbud resonance parameters, the alternative parameters entangle the energy dimension with the wavenumber one: as for the radioactive poles, one must specify on which sheet of the Riemann surface (\ref{eq:rho_c(E) mapping}) are the alternative poles (c.f. theorems 2 of \cite{Ducru_shadow_Brune_Poles_2020}). However, proper analytic continuation will unfold the sheets of Riemann surface (\ref{eq:rho_c(E) mapping}) and thus render such specification useless, as we show in theorem 3 of \cite{Ducru_shadow_Brune_Poles_2020} --- this is another strong argument in favor of analytic continuation of R-matrix operators, in particular the shift $S_c(\rho_c)$ and penetration $P_c(\rho_c)$ functions (contrarily to ``force-closure'' legacy of Lane \& Thomas).
Moreover, in practice this is not as much of a limitation, as we showed in theorem 4 of \cite{Ducru_shadow_Brune_Poles_2020} that we can always choose the first $N_\lambda$ resonant alternative poles of the physical sheet $\left\{E,+\right\}$.
Nonetheless, all Reich-Moore and sub-threshold alternative parameters still change depending on whether the shift $S_c(E)$ and penetration $P_c(E)$ functions are analytically continued (theorem 3 of \cite{Ducru_shadow_Brune_Poles_2020}) or ``forced-closed'' as defined by Lane \& Thomas (theorem 2 of \cite{Ducru_shadow_Brune_Poles_2020}): we here argue that the physically and mathematically correct way is to perform analytic continuation of the shift $S_c(\rho_c)$ and penetration $P_c(\rho_c)$ functions, and provide many more arguments for this in section \ref{sec::Scattering matrix continuation to complex energies}. 

Note that a commonly alleged advantage of the alternative poles $\widetilde{E_\lambda}$ is that they correspond to the peak of the resonances -- actually of the Kapur-Peierls operator $\boldsymbol{R}_L(E)$ since the cross section has an additional $\frac{1}{\left|k_c(E)\right|^2}$ modulating term (see \cite{Ducru_WMP_Phys_Rev_C_2019} for more discussion on this). Though this is true in the case of full R-matrix equations (where the resonance energies are real) for resonant poles above threshold (not the shadow poles discovered in \cite{Ducru_shadow_Brune_Poles_2020}), this ceases to be true for channel-eliminated Reich-Moore evaluations (where the resonance energies are in effect complex $E_\lambda - \mathrm{i}\frac{\Gamma_{\lambda,\gamma}}{2}$ as explained in section II.A.4. of \cite{Ducru_WMP_Phys_Rev_C_2019}). Indeed, the alternative poles $\widetilde{E_\lambda}$ are then complex (c.f. section IV.A of \cite{Ducru_shadow_Brune_Poles_2020}), and neither analytic continuation nor Lane \& Thomas force-closure will entail their real parts exactly correspond to the Kapur-Peierls operator $\boldsymbol{R}_L(E)$ resonance peaks. 
The exact peaks of the $\boldsymbol{R}_L(E)$ resonances are actually the real parts of the radioactive poles $\Re\left[ \mathcal{E}_j\right]$, and the widths are the imaginary parts $\Im\left[ \mathcal{E}_j\right]$, which we here document in table \ref{tab:Xe-134 radioactive parameters} and shown in figure \ref{fig:xenon-134 J=1/2(-) Kapur-Peierls operator} for the two p-wave resonances of xenon $^{\mathrm{134}}\mathrm{Xe}$ spin-parity group $J^\pi = 1/2^{(-)}$.  
In practice, though, the real part of the alternative poles $\Re\left[\widetilde{E_\lambda}\right]$ are close (but not identical) to the real part of the radioactive poles $\Re\left[\mathcal{E}_j\right]$ (one needs to go to more digits to see the discrepancy between values of table VI in \cite{Ducru_shadow_Brune_Poles_2020} to our table \ref{tab:Xe-134 radioactive parameters} here), and as such are much closer to the peak of the resonances than are the Wigner-Eisenbud resonance energies $E_\lambda$.

Another important characteristic of the radioactive parameters is that they are the bridge between the R-matrix theory parametrizations of nuclear reactions, and the scattering matrix pole expansions of Humblet and Rosenfeld, as we now explain in section \ref{sec: Radioactive parameters link R-matrix theory to the scattering matrix pole expansions}.

\subsection{\label{sec: Radioactive parameters link R-matrix theory to the scattering matrix pole expansions} Radioactive parameters link R-matrix theory to the scattering matrix pole expansions}

So far, we have started from the R-matrix Wigner-Eisenbud parameters $\left\{E_\lambda, \gamma_{\lambda,c}\right\}$ to construct the poles and residues of the Kapur-Peierls operator $\boldsymbol{R}_{L}$, through (\ref{eq:R_L radioactive problem}) and (\ref{eq:R_L residues widths}).
We here show that these Siegert-Humblet radioactive parameters are the link between R-matrix theory (c.f. \cite{Wigner_and_Eisenbud_1947,Bloch_1957,Lane_and_Thomas_1958}) and the scattering matrix pole expansions of Humblet-Rosenfeld and others (c.f. \cite{Dyatlov, Guillope_1989_ENS, S-matrix-complex_resonance_parameters_Csoto_1997, Theory_of_Nuclear_Reactions_I_resonances_Humblet_and_Rosenfeld_1961, Theory_of_Nuclear_Reactions_II_optical_model_Rosenfeld_1961, Theory_of_Nuclear_Reactions_III_Channel_radii_Humblet_1961_channel_Radii, Theory_of_Nuclear_Reactions_IV_Coulomb_Humblet_1964, Theory_of_Nuclear_Reactions_V_low_energy_penetrations_Jeukenne_1965, Theory_of_Nuclear_Reactions_VI_unitarity_Humblet_1964, Theory_of_Nuclear_Reactions_VII_Photons_Mahaux_1965, Theory_of_Nuclear_Reactions_VIII_evolutions_Rosenfeld_1965}). 

Indeed, plugging-in the Kapur-Peierls $\boldsymbol{R}_L$ operator expansion (\ref{eq::RL Mittag Leffler}) into the expression of the scattering matrix (\ref{eq:U expression}) then yields the Mittag-Leffler expansion of the scattering matrix:
\begin{equation}
\boldsymbol{U}(E) \underset{\mathcal{W}(E)}{=} \boldsymbol{w}\sum_{j\geq 1} \frac{\boldsymbol{u_j}\boldsymbol{u_j}^\mathsf{T}}{E - \mathcal{E}_j} + \boldsymbol{\mathrm{Hol}}_{\boldsymbol{U}}(E)
\label{eq::U Mittag Leffler}
\end{equation}
where $\boldsymbol{w} \triangleq 2\mathrm{i}\Id{}$ is the Wronskian (\ref{eq:wronksian expression}), and the scattering residue widths $\boldsymbol{u_j}$ are defined as:
\begin{equation}
\boldsymbol{u_j} \triangleq \left[\boldsymbol{\rho}^{1/2}\boldsymbol{O}^{-1}\right]_{E = \mathcal{E}_j}\boldsymbol{r_j}
\label{eq::u_j scattering residue width}
\end{equation}
In writing (\ref{eq::U Mittag Leffler}), we have used the fact that all the resonances of the scattering matrix $\boldsymbol{U}(E)$ come from the Kapur-Peierls radioactive poles $\left\{\mathcal{E}_j\right\}$ -- indeed, we demonstrate in theorem \ref{theo: Scattering matrix poles are the Siegert-Humblet radioactive poles}, section \ref{subsec::Spurious poles cancellation for analytically continued scattering matrix}, that the poles $\left\{\omega_n\right\}$ of the outgoing wave function $\boldsymbol{O}(E)$ cancel out in (\ref{eq:U expression}) and are thus not present in the scattering matrix.
Cauchy's residues theorem then allows us to evaluate the residues at the pole value to obtain (\ref{eq::u_j scattering residue width}).
As for (\ref{eq::RL Mittag Leffler}), if a resonance were to be degenerate with multiplicity $M_j$, the residues would no longer be rank-one, but instead the scattering matrix residue associated to pole $\mathcal{E}_j$ would be $\sum_{m=1}^{M_j}\boldsymbol{u_j^m}{\boldsymbol{u_j^m}}^\mathsf{T}$, with $\boldsymbol{u_j^m} \triangleq \left[\boldsymbol{\rho}^{1/2}\boldsymbol{O}^{-1}\right]_{E = \mathcal{E}_j}\boldsymbol{r_{j}^m}$.

Expression (\ref{eq::U Mittag Leffler}) exhibits the advantage that the energy dependence of the scattering matrix $\boldsymbol{U}(E)$ is untangled in a simple sum. All the resonance behavior stems from the complex poles and residue widths $\left\{\mathcal{E}_j , u_{j,c} \right\}$, which yield the familiar Breit-Wigner profiles (Cauchy-Lorentz distributions) for the cross section. Conversely, all the threshold behavior and the background are described by the holomorphic part $\boldsymbol{\mathrm{Hol}}_{\boldsymbol{U}}(E)$, which can be expanded in various forms, for instance analytically (\ref{eq::Hol_U expansion}).

This establishes the important bridge between the R-matrix parametrizations and the Humblet-Rosenfeld expansions of the scattering matrix. More precisely, Mittag-Leffler expansion (\ref{eq::U Mittag Leffler}) is identical to the Humblet-Rosenfeld expansions (10.22a)-(10.22b) in \cite{Theory_of_Nuclear_Reactions_I_resonances_Humblet_and_Rosenfeld_1961} for the neutral particles case, and (5.4a)-(5.4b) in \cite{Theory_of_Nuclear_Reactions_IV_Coulomb_Humblet_1964} for the Coulomb case. We thus here directly connect the R-matrix parameters with the Humblet-Rosenfeld resonances, parametrized by their partial widths and real and imaginary poles, as described in \cite{Theory_of_Nuclear_Reactions_III_Channel_radii_Humblet_1961_channel_Radii}.
In particular, the poles $\left\{\mathcal{E}_j\right\}$ from (\ref{eq:E_j pole def}), found by solving (\ref{eq:R_L radioactive problem}), are exactly the ones defined by equations (9.5) and (9.8) in \cite{Theory_of_Nuclear_Reactions_I_resonances_Humblet_and_Rosenfeld_1961}.
The scattering residue widths $\left\{u_{j,c}\right\}$, calculated from (\ref{eq::u_j scattering residue width}), then correspond to the Humblet-Rosenfeld complex residues (10.12) in \cite{Theory_of_Nuclear_Reactions_I_resonances_Humblet_and_Rosenfeld_1961}, from which they build their quantities $\left\{G_{c,n}\right\}$ appearing in expansions (10.22a)-(10.22b) in \cite{Theory_of_Nuclear_Reactions_I_resonances_Humblet_and_Rosenfeld_1961}, or (5.4a)-(5.4b) in \cite{Theory_of_Nuclear_Reactions_IV_Coulomb_Humblet_1964}.
Finally, the holomorphic part $\boldsymbol{\mathrm{Hol}}_{\boldsymbol{U}}(E)$ corresponds to the regular function $Q_{c,c'}(E)$ defined between (10.14a) and (10.14b) in \cite{Theory_of_Nuclear_Reactions_I_resonances_Humblet_and_Rosenfeld_1961}.

Just as Humblet and Rosenfeld did with $Q_{c,c'}(E)$ in section 10.2 of \cite{Theory_of_Nuclear_Reactions_I_resonances_Humblet_and_Rosenfeld_1961} and section 4 of \cite{Theory_of_Nuclear_Reactions_IV_Coulomb_Humblet_1964}, we do not give here an explicit way of calculating this holomorphic contribution $\boldsymbol{\mathrm{Hol}}_{\boldsymbol{U}}(E)$ other than stating that it is possible to expand it in various ways. Far from a threshold, an analytic series in energy space $E$ can stand:
\begin{equation}
\boldsymbol{\mathrm{Hol}}_{\boldsymbol{U}}(E) \underset{\mathcal{W}(E)}{=}  \sum_{n\geq0} \boldsymbol{s}_n E^n
\label{eq::Hol_U expansion}
\end{equation}
In the immediate vicinity of a threshold, the asymptotic threshold behavior will prevail (for massive particles, $U_{c,c'}\sim k_c^{\ell_c + 1} k_{c'}^{\ell_{c'}}$, c.f. eq.(10.5) in \cite{Theory_of_Nuclear_Reactions_I_resonances_Humblet_and_Rosenfeld_1961}, or \cite{Wigner_Thresholds_1948}), yielding an expansion in wavenumber space of the form:
\begin{equation}
\boldsymbol{\mathrm{Hol}}_{\boldsymbol{U}}(E) \underset{\mathcal{W}(E_{T_c})}{=} \sum_{n\geq0} \boldsymbol{s}_n k_c^n(E)
\label{eq::Hol_U threshold expansion}
\end{equation}
Though there is no explicit way of linking these expansions (\ref{eq::Hol_U threshold expansion}) or (\ref{eq::Hol_U expansion}) to the R-matrix Wigner-Eisenbud parameters $\left\{E_\lambda, \gamma_{\lambda, c}\right\}$, this means that the same approach as that discussed in the paragraph following theorem \ref{theo::Siegert-Humblet Radioactive Poles} can be taken: one can provide a local set of coefficients $\left\{\boldsymbol{s}_n\right\}_{\mathcal{W}(E)}$ to expand the holomorphic part of the scattering matrix $\boldsymbol{\mathrm{Hol}}_{\boldsymbol{U}}(E)$, and then calculate the scattering matrix from the Mittag-Leffler expansion (\ref{eq::U Mittag Leffler}). This is at the core of the windowed multipole representation of R-matrix cross sections established in \cite{Ducru_WMP_Phys_Rev_C_2019}.

An important question is that of the radius of convergence of the Mittag-Leffler expansion (\ref{eq::U Mittag Leffler}), in other terms how big can the vicinity $\mathcal{W}(E)$ be?
Humblet and Rosenfeld analyze this problem in section 1.4 of \cite{Theory_of_Nuclear_Reactions_I_resonances_Humblet_and_Rosenfeld_1961}, and perform the Mittag-Leffler expansion (1.50). In the first paragraph of p.538 it is stated that Humblet demonstrated in his Ph.D. thesis that the Mittag-Leffler series will converge for $M\geq 1$ for $U(k)$, though this does not investigate the multi-channel case, and thus the multi-sheeted nature of the Riemann surface stemming from mapping (\ref{eq:rho_c(E) mapping}). They assume at the beginning of section 10.2 that this property stands in the multi-channel case and yet continue their discussion with a choice of $M=0$ that would leave the residues diverging according to their expansion (1.50). This is one reason why we chose in this article to start from a local Mittag-Leffler expansion, and then search for its domain of convergence.
General mathematical scattering theory shows that the Mittag-Leffler expansion holds at least on the whole physical sheet (c.f. theorem 0.2 p.139 of \cite{Guillope_1989_ENS}).
Moreover, in his article on ``threshold behaviour in quantum field theory''\cite{edenThresholdBehaviourQuantum1952}, Eden proves that ``between the threshold values [...] the [Scattering] matrix elements are analytic functions of the energies and momenta of the
incident particles'', though it does not specify in which form the Mittag-Leffler expansion will converge separately on each sheet. 
In practice this requirement is not needed since it is often computationally more advantageous to break down an energy region between two consecutive thresholds $\left[ E_{T_c} , E_{T_{c+1}} \right]$ into smaller vicinities (a compression method for efficient computation used in the windowed multipole library \cite{Ducru_WMP_Phys_Rev_C_2019}).

As we see, by performing the Mittag-Leffler expansion (\ref{eq::U Mittag Leffler}), we have traded-off a finite set of real, unwound, Wigner-Eisenbud parameters $\left\{E_\lambda,\gamma_{\lambda,c}\right\}$ that completely parametrize the energy dependence of the scattering matrix through (\ref{eq:U expression}), with an infinite set of complex Siegert-Humblet radioactive parameters $\left\{\mathcal{E}_j,r_{j,c}\right\}$ plus some local coefficients $\left\{\boldsymbol{s_n}\right\}_{\mathcal{W}(E)}$ for the holomorphic part, all intricately intertwined through radioactive problem (\ref{eq:R_L radioactive problem}), which makes them dwell on a sub-manifold of the multi-sheeted Riemann surface of mapping (\ref{eq:rho_c(E) mapping}). This additional complexity of the Siegert-Humblet parameters comes at the gain of a simple parametrization of the energy dependence for the scattering matrix: the poles and residues expansion (\ref{eq::U Mittag Leffler}). For computational purposes, this may sometimes be a trade-off worth doing: this is the basis for the windowed multipole representation of R-matrix cross sections \cite{Ducru_WMP_Phys_Rev_C_2019}.

\section{\label{sec:Radioactive parameters invariance to channel radii}Radioactive parameters invariance to channel radii}

Section \ref{sec:Siegert-Humblet pole expansion in radioactive states} provided new insights into the link between the Humblet-Rosenfeld scattering matrix pole expansions, and both the Wigner-Eisenbud, and the Siegert-Humblet parametrizations of R-matrix theory. Concerning invariance to arbitrary parameters, we saw that the Siegert-Humblet parameters are invariant under change of boundary condition $B_c$, but not under change of channel radius $a_c$ --- this is also true for the alternative parameters discussed in \cite{Ducru_shadow_Brune_Poles_2020}.
This section is dedicated to invariance properties of the Siegert-Humblet radioactive parameters to a change in channel radius $a_c$. This problem is less studied than that of the invariance to the boundary conditions $B_c$. To the best of our knowledge, the only previous results on this topic are the partial differential equations on the Wigner-Eisenbud $\left\{E_\lambda , \gamma_{\lambda,c}\right\}$ parameters Teichmann derived in his Ph.D. thesis (c.f. eq. (2.29) and (2.31) sections III.2. p.27 of \cite{Teichmann_thesis_1949}), a recent study of the limit case $a_c \to 0$ in \cite{G_Hale_channel_radius_limit_2014_PhysRevC.89.014623}, as well as the general results of the variations of the R-matrix to any arbitrary parameter by Mockel and Perez (c.f equations (71) and (75) \cite{mockelVariationMatrix1970}).
We here focus on the Siegert-Humblet parameters $\big\{\mathcal{E}_j , r_{j,c}\big\}$. Our main result of this section resides in theorem \ref{theo:: Radioactive parameters transformation under change of channel radius}, which establishes a way of converting the Siegert-Humblet radioactive parameters under a change of channel radius $a_c$.

\begin{theorem}\label{theo:: Radioactive parameters transformation under change of channel radius} \textsc{Radioactive parameters transformation under change of channel radius $a_c$.} \\
Let the radioactive poles $\left\{\mathcal{E}_j\right\}$ be the solutions of the radioactive problem (\ref{eq:R_L radioactive problem}).
Under a change of channel radius $a_c^{(0)} \to a_c$ (or infinitesimal $\frac{\partial \cdot }{\partial a_c }$):
\begin{itemize}
    \item the Kapur-Peierls operator $\boldsymbol{R}_{L}$, defined in (\ref{eq: def Kapur-Peierls operator}), is subject to the following partial differential equations: for the diagonal elements,
\begin{equation}
\begin{IEEEeqnarraybox}[][c]{rcl}
a_c \frac{\partial {R_L}_{cc}}{\partial a_c} + (1- 2 L_c) {R_L}_{cc} - 1 & \ = \ &  0
\IEEEstrut\end{IEEEeqnarraybox}
\label{eq: dR//da cc}
\end{equation}
and for off-diagonal ones,
\begin{equation}
\begin{IEEEeqnarraybox}[][c]{rcl}
a_c \frac{\partial {R_L}_{cc'}}{\partial a_c} + (\frac{1}{2}- L_c) {R_L}_{cc'}  & \ = \ &  0
\IEEEstrut\end{IEEEeqnarraybox}
\label{eq: dR//da cc'}
\end{equation}
    \item the radioactive poles $\left\{\mathcal{E}_j\right\}$ are invariant:
    \begin{equation}
\begin{IEEEeqnarraybox}[][c]{rcl}
\frac{\partial \mathcal{E}_j}{\partial a_c} & \ = \ &  0
\IEEEstrut\end{IEEEeqnarraybox}
\label{eq: dp/dac = 0}
\end{equation}
    \item the radioactive widths $\left\{r_{j,c}\right\}$ (widths of the Kapur-Peierls $\boldsymbol{R}_{L}$ operator residues (\ref{eq:R_L residues for degenerate case})), are subject to the following first-order linear partial differential equation:
\begin{equation}
\begin{IEEEeqnarraybox}[][c]{rcl}
a_c \frac{\partial {r}_{j,c}}{\partial a_c} + (\frac{1}{2}- L_c) {r}_{j,c} & \ = \ &  0
\IEEEstrut\end{IEEEeqnarraybox}
\label{eq: drcj/dac }
\end{equation}
 \item which can be formally solved as,
\begin{equation}
\begin{IEEEeqnarraybox}[][c]{rcl}
{r}_{j,c}(a_c) = {r}_{j,c}(a_c^{(0)}) \, \sqrt{\frac{a_c^{(0)}}{a_c}} \, \exp{\left(\int_{a_c^{(0)}}^{a_c} \frac{L_c(k_c x)}{x} \mathrm{d}x\right)}
\IEEEstrut\end{IEEEeqnarraybox}
\label{eq: rcj(ac) integral form}
\end{equation}
\item and explicitly integrates to:
\begin{equation}
\begin{IEEEeqnarraybox}[][c]{rcl}
\frac{{r}_{j,c}(a_c)}{{r}_{j,c}(a_c^{(0)})} =  \frac{O_c(\rho_c(a_c))}{O_c(\rho_c(a_c^{(0)}))} \sqrt{\frac{a_c^{(0)}}{a_c}}
\IEEEstrut\end{IEEEeqnarraybox}
\label{eq: rcj(ac) explicit}
\end{equation}
\item Moreover, letting $\left\{ \omega_n \right\}$ be the roots of the outgoing wave function $\Big\{ \omega_n \; | \; O_c(\omega_n) = 0 \Big\}$, the latter (\ref{eq: rcj(ac) explicit}) can take the following elemental product expansion: 
\begin{equation}
\begin{IEEEeqnarraybox}[][c]{rcl}
\frac{{r}_{j,c}(a_c)}{{r}_{j,c}(a_c^{(0)}) } = \sqrt{\frac{a_c^{(0)}}{a_c}} \left( \frac{ a_c^{(0)} }{  a_c}\right)^{\ell} \mathrm{e}^{\mathrm{i}k_c\left(a_c - a_c^{(0)}\right)} \prod_{n\geq1}\left( \frac{k_c a_c - \omega_n}{k_c a_c^{(0)} - \omega_n}\right)
\IEEEstrut\end{IEEEeqnarraybox}
\label{eq: rcj(ac) - integrated by Mittag-Leffler}
\end{equation}
where there are an infinite number of such roots $\left\{ \omega_n \right\}$ in the Coulomb case, while for neutral particle channel $c $ with angular momentum $\ell$, there exists exactly $\ell$ roots $\left\{ \omega_n \right\}_{n\in\llbracket 1, \ell \rrbracket}$, the exact and algebraically solvable values of which are reported, up to angular momentum $\ell = 4$, in table II of \cite{Ducru_shadow_Brune_Poles_2020}.

\end{itemize}

\end{theorem}

\begin{proof}
We start by bringing forth the observation that the scattering matrix $\boldsymbol{U}$ is invariant under change of channel radius $a_c$, i.e. for any channel $c$ we have:
\begin{equation}
\begin{IEEEeqnarraybox}[][c]{rcl}
\frac{\partial \boldsymbol{U}}{\partial a_c}  & \ = \ &  \boldsymbol{0}
\IEEEstrut\end{IEEEeqnarraybox}
\label{eq: dU/dac = 0}
\end{equation}
Since theorem \ref{theo: Scattering matrix poles are the Siegert-Humblet radioactive poles} will show that the poles of the scattering matrix are exactly the ones of the Kapur-Peierls operator $\boldsymbol{R}_{L}$, which are the Siegert-Humblet poles $\left\{\mathcal{E}_j\right\}$, invariance (\ref{eq: dU/dac = 0}) entails that the radioactive poles are invariant under change of channel radius $a_c$, i.e. (\ref{eq: dp/dac = 0}).

This is not the case for the radioactive widths $\left\{ r_{j,c}\right\}$. However, one can use invariance (\ref{eq: dU/dac = 0}) to differentiate the scattering matrix $\boldsymbol{U}$ expression (\ref{eq:U expression}).
The $\boldsymbol{L}$ operator definition (\ref{eq: L operator}), and $\rho_c = k_c a_c$, entail
\begin{equation}
\begin{IEEEeqnarraybox}[][c]{rcl}
\frac{\partial \rho_c^{1/2}O_c^{-1} }{\partial a_c}  & \ = \ & \frac{1}{a_c}\rho_c^{1/2}O_c^{-1}\left[ \frac{1}{2} - L_c\right]
\IEEEstrut\end{IEEEeqnarraybox}
\label{eq: drhoO-1/dac = 0}
\end{equation}
this enables us to establish the partial differential equations (\ref{eq: dR//da cc}) and (\ref{eq: dR//da cc'}) on the Kapur-Peierls matrix operator $\boldsymbol{R}_{L}$ elements, which can be synthesized into expression,
\begin{equation}
\begin{IEEEeqnarraybox}[][c]{rcl}
\boldsymbol{a} \frac{\partial \boldsymbol{R}_{L}}{\partial \boldsymbol{a}} + ( \frac{1}{2} \Id{} - \boldsymbol{L}) \boldsymbol{R}_{L} + \Id{}\circ\left[ (\frac{1}{2}\Id{} - \boldsymbol{L}) \boldsymbol{R}_{L} - \Id{}  \right] & \ = \ &  \boldsymbol{0}
\IEEEstrut\end{IEEEeqnarraybox}
\label{eq: dRL/da }
\end{equation}
where $\circ$ designates the Hadamard matrix product, and where we used the notation:
\begin{equation}
\begin{IEEEeqnarraybox}[][c]{rcl}
\left[\frac{\partial \boldsymbol{R}_{L}}{\partial \boldsymbol{a}}\right]_{cc'}  & \ \triangleq \ &  \frac{\partial {R_L}_{cc'}}{\partial a_c}
\IEEEstrut\end{IEEEeqnarraybox}
\label{eq: dRL/da def}
\end{equation}
Equivalently, inverting the Kapur-Peierls operator in differential equation (\ref{eq: dRL/da def}) yields the following Riccati equation:
\begin{equation}
\begin{IEEEeqnarraybox}[][c]{rcl}
\boldsymbol{a} \frac{\partial \boldsymbol{R}_{L}^{-1}}{\partial \boldsymbol{a}} - \boldsymbol{R}_{L}^{-1} ( \frac{1}{2} \Id{} - \boldsymbol{L})  - \Id{}\circ\left[ \boldsymbol{R}_{L}^{-1}(\frac{1}{2}\Id{} - \boldsymbol{L})  - \boldsymbol{R}_{L}^{-2}  \right] & \ = \ &  \boldsymbol{0}
\IEEEstrut\end{IEEEeqnarraybox}
\label{eq: dinvRL/da }
\end{equation}
These first order partial differential equations on the Kapur-Peierls operator $\boldsymbol{R}_L$ are equivalent to relations (71) and (75) Mockel and Perez established for the $\boldsymbol{R}$ matrix in \cite{mockelVariationMatrix1970}. They are quite inconvenient to solve in that they are channel-dependent, and thus give rise to equations for each cross term.
Remarkably, this is not the case for the radioactive residues. 

Having demonstrated the radioactive poles invariance (\ref{eq: dp/dac = 0}), Mittag-Leffler expansion (\ref{eq::U Mittag Leffler}) entails that $\boldsymbol{u_j}$ from (\ref{eq::u_j scattering residue width}) satisfies invariance: $\frac{\partial \boldsymbol{u_j}}{\partial a_c}   \ = \   \boldsymbol{0}$.
Applying result (\ref{eq: drhoO-1/dac = 0}) to the latter then yields partial differential equation (\ref{eq: drcj/dac }), the direct integration of which readily yields (\ref{eq: rcj(ac) integral form}). Since $L_c(\rho_c) \triangleq \frac{\rho_c}{O_c(\rho_c)}\frac{\partial O_c(\rho_c)}{\partial \rho_c}$, (\ref{eq: rcj(ac) integral form}) integrates explicitly to (\ref{eq: rcj(ac) explicit}). This result also stands for any degenerate state of multiplicity $M_j$, where for each radioactive width $\boldsymbol{r_j^m}$ we have:
\begin{equation}
\begin{IEEEeqnarraybox}[][c]{rcl}
\frac{{r}_{j,c}^m(a_c)}{{r}_{j,c}^m(a_c^{(0)})} =  \frac{O_c(\rho_c(a_c))}{O_c(\rho_c(a_c^{(0)}))} \sqrt{\frac{a_c^{(0)}}{a_c}}
\IEEEstrut\end{IEEEeqnarraybox}
\label{eq: rcj(ac) explicit degenerate m}
\end{equation}

Finally, the proof of (\ref{eq: rcj(ac) - integrated by Mittag-Leffler}) is the element-wise integration of (\ref{eq: rcj(ac) integral form}) using the Mittag-Leffler pole expansion (\ref{eq::Explicit Mittag-Leffler expansion of L_c}) of $L_c(\rho)$, which we established in theorem 1 of \cite{Ducru_shadow_Brune_Poles_2020} -- invoking Fubini's theorem to permute sum and integral.
In the case of neutral particles, there is a finite number of roots $\left\{ \omega_n \right\}$ so that the product in (\ref{eq: rcj(ac) - integrated by Mittag-Leffler}) is finite.
Note that in the charged particles case, there is an infinite number (countable) of roots $\left\{\omega_n\right\}$, and the Weierstrass factorization theorem would thus usually require (\ref{eq: rcj(ac) - integrated by Mittag-Leffler}) to be cast in a Hadamard canonical representation with Weierstrass elementary factors. However, in (\ref{eq: rcj(ac) - integrated by Mittag-Leffler}), the product elements tend towards unity as $n$ goes to infinity $\left( \frac{k_c a_c - \omega_n}{k_c a_c^{(0)} - \omega_n}\right) \underset{n\to \infty}{\longrightarrow} 1$, so that the infinite product in (\ref{eq: rcj(ac) - integrated by Mittag-Leffler}) should still converge.

\end{proof}

Note that for neutral particles (massive or massless) s-waves ($\ell=0$), the outgoing wave function is $O_c(\rho(a_c)) = \mathrm{e}^{\mathrm{i}k_ca_c}$ (c.f. table I of \cite{Ducru_shadow_Brune_Poles_2020}), so that (\ref{eq: rcj(ac) explicit}) yields ${r}_{j,c}(a_c) = {r}_{j,c}(a_c^{(0)}) \, \sqrt{\frac{a_c^{(0)}}{a_c}} \, \mathrm{e}^{\mathrm{i}k_c\left(a_c - a_c^{(0)}\right)}$. Alternatively, directly integrating (\ref{eq: rcj(ac) integral form}) with the outgoing-wave reduced logarithmic derivative expression $L_c(\rho(a_c)) = \mathrm{i}k_c a_c$ yields back the same result. Thus for s-wave neutral channels subject to a change of channel radius, the modulus of the radioactive widths decreases proportionally to the inverse square root of the channel radius $a_c$, at least for real wavenumbers $k_c \in \mathbb{R}$, i.e. real energies above the channel threshold.
Since the transition probability rates partial widths can be defined as the square of the modulus of the radioactive width (c.f. eq. (6) in \cite{Hale_1987}), this means these transition partial widths decrease inversely to the channel radius:
$\left|\frac{{r}_{j,c}(a_c)}{{r}_{j,c}(a_c^{(0)})} \right|^2=  \frac{a_c^{(0)}}{a_c}$.

A striking property of the R-matrix parametrizations is that they separate the channel contribution to each resonance, meaning that to compute, for instance, the $R_{c,c'}$ element in (\ref{eq:R expression}), one only requires the widths for each level of each channel, $\gamma_{\lambda,c}$, and not some new parameter for each specific channel pair $c,c'$ combination.
In this spirit, we show in theorem \ref{theo:: Radioactive parameters transformation under change of channel radius} that the Siegert-Humblet radioactive widths $r_{j,c}$ play a similar role in that their transformation under a change of channel radius only depends on that given channel.

Theorem \ref{theo:: Radioactive parameters transformation under change of channel radius} makes explicit the behavior of the radioactive widths $\big\{r_{j,c}\big\}$ under a change of channel radius $a_c$.
Strikingly, only the Kapur-Peierls matrix $\boldsymbol{R}_{L}$ appears in this change of variable. 
This means that the R-matrix $\boldsymbol{R}$ and the $\boldsymbol{L^0}$ matrix function suffice to both compute the Siegert-Humblet parameters $\big\{\mathcal{E}_j, r_{j,c}\big\}$ from (\ref{eq:R_L radioactive problem}), and to change the radioactive widths $\big\{r_{j,c}\big\}$ under a change of channel radius $a_c$.
This novel result portrays the Siegert-Humblet parameters as allowing a simple energy dependence to the scattering matrix (\ref{eq::U Mittag Leffler}) --- albeit locally and needing the expansion coefficients (\ref{eq::Hol_U expansion}) --- all the while being boundary condition $B_c$ independent and easy to transform under a change of channel radius $a_c$.

\section{\label{sec::Scattering matrix continuation to complex energies} Scattering matrix continuation to complex energies}

In section 5.2 of \cite{Theory_of_Nuclear_Reactions_I_resonances_Humblet_and_Rosenfeld_1961}, Humblet and Rosenfeld continue the scattering matrix to complex wave numbers $k_c \in \mathbb{C}$, and define corresponding open and closed channels. They however never point to the conundrum that this entails: in their approach, the scattering matrix seemingly does not annul itself below threshold.
This is contrary to the approach taken by Lane \& Thomas, where they explicitly annul the elements of the scattering matrix below thresholds, as stated in the paragraph between equations (2.1) and (2.2) of section VII.1. p.289 \cite{Lane_and_Thomas_1958}.
Claude Bloch ingeniously circumvents the problem by explicitly stating after eq. (50) in \cite{Bloch_1957} that the scattering matrix is a matrix of the open channels only, meaning its dimensions change as more channels open when energy $E$ increases past new thresholds $E > E_{T_c}$. In his approach, sub-threshold elements of the scattering matrix need not be annulled, one simply does not consider them.

We dedicate this section to this question of how to extend the scattering matrix to complex wavenumbers $k_c\in\mathbb{C}$, while closing the channels below threshold.
We argue that analytic continuation of R-matrix operators (lemma 2 section III.B of \cite{Ducru_shadow_Brune_Poles_2020}) is the physically correct way of constructing the scattering matrix for complex wavenumbers.
To support this, we advance and demonstrate three new arguments: analytic continuation cancels out spurious poles otherwise introduced by the outgoing wavefunctions $O_c$ (theorem \ref{theo: Scattering matrix poles are the Siegert-Humblet radioactive poles}); analytic continuation respects generalized unitarity (theorem \ref{theo::satisfied generalized unitarity condition}); and, for massive particles (not photons), analytic continuation of real wavenumber expressions to sub-threshold energies naturally sees the transmission matrix evanesce on the physical sheet (theorem \ref{theo::evanescence of sub-threshold transmission matrix}), while always closing the channels by annulling the cross section (theorem \ref{theo::Analytic continuation annuls sub-threshold cross sections}).

\subsection{\label{subsec::Forcing sub-threshold elements to zero} Forcing sub-threshold elements to zero: the legacy of Lane \& Thomas}

To close the channels for real energies below threshold, the simplest approach is the one proposed by Lane \& Thomas in \cite{Lane_and_Thomas_1958}.
The scattering matrix expressions (\ref{eq:U expression}) can be re-written, for real energies above threshold, according to section VII.1 equation (1.6b) in \cite{Lane_and_Thomas_1958}:
\begin{equation}
\begin{IEEEeqnarraybox}[][c]{rcl}
\boldsymbol{U} & \ = \ &  \boldsymbol{\Omega} \left( \Id{}+ \boldsymbol{w} \boldsymbol{\mathfrak{P}}^{1/2}\boldsymbol{R}_{L}\boldsymbol{\mathfrak{P}}^{1/2}  \right) \boldsymbol{\Omega}
\IEEEstrut\end{IEEEeqnarraybox}
\label{eq:U expression L&T}
\end{equation}
with Wronskian $\boldsymbol{w}$ from (\ref{eq:wronksian expression}) and the values defined for energies above the thresholds in III.3.a. p.271 of \cite{Lane_and_Thomas_1958}:
\begin{equation}
\begin{IEEEeqnarraybox}[][c]{rcl}
\boldsymbol{\Omega} & \ \triangleq \ & \boldsymbol{O}^{-1}\boldsymbol{I} \\
\boldsymbol{\mathfrak{P}} & \ \triangleq \ & \boldsymbol{\rho} \boldsymbol{O}^{-1}\boldsymbol{I}^{-1}
\IEEEstrut\end{IEEEeqnarraybox}
\label{eq:Omega and mathfrak P expression L&T}
\end{equation}
Let us note that the Mittag-Leffler expansion (\ref{eq::RL Mittag Leffler}) of the Kapur-Peierls matrix $\boldsymbol{R}_{L}$ operator can still be performed.

The Lane \& Thomas ``sub-threshold channel force-closure'' approach exploits the ambiguity in the definition of the shift $\boldsymbol{S}(E)$ and penetration $\boldsymbol{P}(E)$ factors:
\begin{equation}
\boldsymbol{L} \ = \ \boldsymbol{S} + i \boldsymbol{P}
\label{eq:: L = S + iP}
\end{equation}
for complex energies $E \in \mathbb{C}$, as discussed in section III.B of \cite{Ducru_shadow_Brune_Poles_2020} (of which we follow the notation).
Lane \& Thomas choose the branch-point definitions for the shift $\boldsymbol{S}$ and penetration $\boldsymbol{P}$ functions, made explicit in lemma 1 section III.B of \cite{Ducru_shadow_Brune_Poles_2020}.
Lane \& Thomas do not specify how they would continue the quantities (\ref{eq:Omega and mathfrak P expression L&T}) for negative energies, as they state ``we need not be concerned with stating similar relations for the negative energy channels" (c.f. paragraph after equation (4.7c), p.271.), but they do specify that $\boldsymbol{P} = \boldsymbol{0}$ below threshold energies and $\boldsymbol{P} = \boldsymbol{\mathfrak{P}}$ above.
This means that plugging-in $\boldsymbol{P} $ in place of $\boldsymbol{\mathfrak{P}}$ in (\ref{eq:Omega and mathfrak P expression L&T}) has the convenient property of automatically closing the reaction channels below threshold, since in that case $U_{c,c'} = \Omega_c \Omega_{c'}$, which annuls the off-diagonal terms of the cross section (the reaction channels $c\neq c'$) when plugged into equation (1.10) in \cite{Lane_and_Thomas_1958} VIII.1. p.291.
Note that this approach only annuls the off-diagonal terms of the scattering cross section, leaving non-zero cross sections for the diagonal $\sigma_{cc}(E)$, even below threshold.  Indeed, equation (4.5a) section III.4.a., p.271 of \cite{Lane_and_Thomas_1958} gives $\Omega_c = \mathrm{e}^{\mathrm{i}(\omega_c - \phi_c)}$, whilst the cross section is begotten by the amplitudes of the \textit{transmission matrix} $\boldsymbol{T}(E)$, defined as $T_{cc'} \triangleq \delta_{cc'}\mathrm{e}^{2\mathrm{i}\omega_c} - U_{cc'} $ in (2.3), section VIII.2., p.292. For sub-threshold real energies, the diagonal term of the transmission matrix is thus equal to $T_{cc} = \mathrm{e}^{2\mathrm{i}\omega_c}\left( 1 - \mathrm{e}^{-2\mathrm{i}\phi_c} \right)$. This means that in the Lane \& Thomas approach, all channels $c'\neq c$ are force-closed to zero below the incoming channel threshold $E < E_{T_c}$, except for the $c \to c$ reaction, which is tactfully overlooked as non-physical.

Of course, this approach comes at the cost of sacrificing the analytic properties of the scattering matrix $\boldsymbol{U}$: since $P_c = \Im\left[L_c\right]$, the penetration factor is no longer meromorphic and thus neither is $\boldsymbol{U}$.
This entails that in decomposition (\ref{eq:U expression L&T}) of the scattering matrix, if one ``force-closes'' the channels using the branch-point definition of Lane \& Thomas --- instead of analytically continuing both $\boldsymbol{\mathfrak{P}}$ and $\boldsymbol{\Omega}$ to complex wavenumbers $\boldsymbol{\rho} \in \mathbb{C}$ --- the scattering matrix $\boldsymbol{U}(E)$ cannot have poles, as there is then no mathematical meaning to such notion.
This goes directly against a vast amount of literature on the analytic properties of the scattering matrix \cite{Eden_and_Taylor, Theory_of_Nuclear_Reactions_I_resonances_Humblet_and_Rosenfeld_1961, Theory_of_Nuclear_Reactions_II_optical_model_Rosenfeld_1961, Theory_of_Nuclear_Reactions_III_Channel_radii_Humblet_1961_channel_Radii, Theory_of_Nuclear_Reactions_IV_Coulomb_Humblet_1964, Theory_of_Nuclear_Reactions_V_low_energy_penetrations_Jeukenne_1965, Theory_of_Nuclear_Reactions_VI_unitarity_Humblet_1964, Theory_of_Nuclear_Reactions_VII_Photons_Mahaux_1965, Theory_of_Nuclear_Reactions_VIII_evolutions_Rosenfeld_1965, Theory_of_Nuclear_Reactions_IX_few_levels_approx_Mahaux_1965, Humblet_1990, Pacific_Journal_Mittag_Leffler_and_spectral_theory_1960, Guillope_1989_ENS, Dyatlov}.
This is the approach presently taken by the SAMMY code at Oak Ridge National Laboratory \cite{SAMMY_2008}, and upon which rest numerous ENDF evaluations \cite{ENDFBVIII8th2018brown}.


We would like to note that under careful reading, this might not actually have been the approach intended by Lane \& Thomas in \cite{Lane_and_Thomas_1958}. Indeed, Lane \& Thomas never specify how to prolong the $\boldsymbol{\mathfrak{P}}$ to sub-threshold energies, and in equation (\ref{eq:U expression L&T}) it is $\boldsymbol{\mathfrak{P}}$ that is present and not $\boldsymbol{P}$.
They do however note in the paragraph between equations (2.1) and (2.2) of section VII.1. p.289, that ``as there are no physical situations in which the $I_c^-$ occur, the components of the [scattering matrix] are not physically significant and one might as well set them equal to zero as can be seen from (1.6b). This may be accomplished without affecting the [positive energy channels] by setting the negative energy components of the Wronskian matrix to zero; $w_c^-=0$. (This means that the $O_c^-$ and $I_c^-$ are not linearly independent.)".
The choice of wording is here important. Indeed, it says that it is possible to set the Wronskian to zero to close channels below the threshold, though it is not necessary. This is yet another way of closing sub-threshold channels that would allow to keep the analytic properties of the scattering matrix, with $\boldsymbol{\mathfrak{P}} \triangleq  \boldsymbol{\rho} \boldsymbol{O}^{-1}\boldsymbol{I}^{-1} $ still analytically continued, albeit at the cost of not knowing when in the complex plane should the Wronskian $w_c$ be set to zero --- perhaps only on $\mathbb{R}_-$, which would then become a branch line. \\
We show in theorem \ref{theo: Scattering matrix poles are the Siegert-Humblet radioactive poles} (section \ref{subsec::Spurious poles cancellation for analytically continued scattering matrix}) that as long as the Wronskian relation (\ref{eq:wronksian expression}) is guaranteed, the poles of the outgoing scattering wave function $O_c$ cancel out of the scattering expressions (\ref{eq:U expression}) and (\ref{eq:U expression L&T}). The Wronskian condition (\ref{eq:wronksian expression}) is conserved when keeping $\boldsymbol{\mathfrak{P}} $ from (\ref{eq:Omega and mathfrak P expression L&T}) analytically continued --- instead of the definition $P_c \triangleq \Im\left[L_c\right]$, which cannot respect the Wronskian relation (\ref{eq:wronksian expression}) --- so that this approach of setting the Wronskian to zero below threshold while analytically continuing the penetration and shift factors would indeed cancel out the spurious poles of the outgoing wave functions $O_c$.

\subsection{\label{subsec::Analytic continuation of the scattering matrix} Analytic continuation of the scattering matrix}

In opposition to the Lane \& Thomas approach, an entire field of physics and mathematics has studied the analytic continuation of the scattering matrix to complex wavenumbers $k_c\in \mathbb{C}$ \cite{Dyatlov, Guillope_1989_ENS, S-matrix-complex_resonance_parameters_Csoto_1997, Favella_and_Reineri_1962, complex_energy_above_inoization_threshold_1969, Theory_of_Nuclear_Reactions_I_resonances_Humblet_and_Rosenfeld_1961, Theory_of_Nuclear_Reactions_II_optical_model_Rosenfeld_1961, Theory_of_Nuclear_Reactions_III_Channel_radii_Humblet_1961_channel_Radii, Theory_of_Nuclear_Reactions_IV_Coulomb_Humblet_1964, Theory_of_Nuclear_Reactions_V_low_energy_penetrations_Jeukenne_1965, Theory_of_Nuclear_Reactions_VI_unitarity_Humblet_1964, Theory_of_Nuclear_Reactions_VII_Photons_Mahaux_1965, Theory_of_Nuclear_Reactions_VIII_evolutions_Rosenfeld_1965}.

As we saw, there is no ambiguity as to how to continue the $\boldsymbol{L}(\boldsymbol{\rho})$ matrix function to complex wave numbers (c.f. theorem 1 section II.B of \cite{Ducru_shadow_Brune_Poles_2020}), and thus the $\boldsymbol{R}_{L}$ Kapur-Peierls operator (\ref{eq: def Kapur-Peierls operator}). Indeed, the incoming $I_c(\rho_c)$ and outgoing $O_c(\rho_c)$ wave functions can be analytically continued to complex wavenumbers $k_c\in \mathbb{C}$ (c.f. theorem 1 section II.B of \cite{Ducru_shadow_Brune_Poles_2020}), and through the multi-sheeted mapping (\ref{eq:rho_c(E) mapping}) to complex energies $E\in \mathbb{C}$. This naturally yields the meromorphic continuation of the scattering matrix to complex energies (\ref{eq::U Mittag Leffler}).
Since many evaluations are performed using decomposition (\ref{eq:U expression L&T}), in practice performing analytic continuation of R-matrix operators thus means continuing (\ref{eq:Omega and mathfrak P expression L&T}) operators $\boldsymbol{\Omega}$ and $\boldsymbol{\mathfrak{P}}$, setting $\boldsymbol{\mathfrak{P}} \triangleq \boldsymbol{P}$, and defining the shift $S_c(\rho_c)$ and penetration $P_c(\rho_c)$ functions as analytically continued complex meromorphic functions (that is definition (44) and lemma 2 of \cite{Ducru_shadow_Brune_Poles_2020} as opposed to the Lane \& Thomas ``force closure'' definition (41) and lemma 1 of \cite{Ducru_shadow_Brune_Poles_2020}).

The shortcoming of this analytic continuation approach is that it does not evidently annul the channel elements of the scattering matrix for sub-threshold energies $E<E_{T_c}$. Indeed, analytic continuation (\ref{eq::U Mittag Leffler}) means the scattering matrix $\boldsymbol{U}$ is a meromorphic operator from $\mathbb{C}$ to $\mathbb{C}$ on the multi-sheeted Riemann surface of mapping (\ref{eq:rho_c(E) mapping}). Unicity of the analytic continuation then means that if the scattering matrix elements are zero below their threshold, $U_{c,c'}(E) = 0 \; , \; \forall E - E_{T_c} \in \mathbb{R}_-$, then it is identically zero for all energies on that sheet of the manifold: $U_{c,c'}(E) = 0 \; ,\; \forall E \in \mathbb{C}$. Thus, the analytic continuation formalism cannot set elements of the scattering matrix to be identically zero below thresholds $\left\{E_{T_c}\right\}$.

This apparent inability to close channels below thresholds is the principal reason why the nuclear data community has stuck to the legacy approach of Lane \& Thomas (lemma 1 section III.B of \cite{Ducru_shadow_Brune_Poles_2020}), when computing the scattering matrix in equation (\ref{eq:U expression}). This has been the subject of an ongoing controversy in the field on how to continue the scattering matrix to complex wave numbers.

\subsection{\label{subsubsec:: Semi-simple poles in R-matrix theory} Assuming semi-simple poles in R-matrix theory}

Before advancing our analytic continuation arguments of channel closure (section \ref{subsec::Closure of sub-threshold cross sections through analytic continuation}) and generalized unitarity (section \ref{subsec::Generalized unitarity}), let us first start with a general note on high-order poles in R-matrix theory (see the consequences for analytic continuation in section \ref{subsec::Spurious poles cancellation for analytically continued scattering matrix}).
Being a high-order pole, as opposed to a simple pole, can bear various meanings. In our context, the three following definitions are of interest: a) \textit{Laurent order}: the order of the polar expansion in the Laurent development in the vicinity of a pole; b) \textit{Algebraic multiplicity}: the multiplicity of the root of the resolvant at a pole value; c) \textit{Geometric multiplicity}: the dimension of the associated nullspace.

From equation (\ref{eq::RL Mittag Leffler degenerate state}) and throughout the article, we have treated the case of degenerate states where the geometric multiplicity $M_j > 1$ was higher than one, leading to rank-$M_j$ residues.
We have however always assumed the Laurent order to be one: in equation (\ref{eq::RL Mittag Leffler degenerate state}), the residues might be rank-$M_j$, but the Laurent order is still unity (no $\frac{1}{(E-\mathcal{E}_j)^2}$ or higher Laurent orders).

In the general case, the Laurent order is greater than one but it does not equal geometric or algebraic multiplicity.
In terms of Jordan normal form, if the Jordan cells had sizes $n_1, ..., n_{m_g}$, then the geometric multiplicity is equal to $m_g$, the algebraic multiplicity $m_a$ is the sum $ m_a = n_1 + ... + n_{m_g}$, and the Laurent order is the maximum $\mathrm{max}\left\{n_1, ..., n_m\right\}$.

Alternatively, these can be defined as follows:
Let $\boldsymbol{M}(z)$ be a complex-symmetric meromorphic matrix operator, with a root at $z=z_0$ (i.e. $\boldsymbol{M}(z_0)$ is non-invertible).
The algebraic multiplicity $m_a$ is the first non-zero derivative of the determinant, i.e. the first integer $m_a\in\mathbb{N}$ such that $\left.\frac{\mathrm{d}^{m_a}}{\mathrm{d}z^{m_a}}\mathrm{det}\Big( \boldsymbol{M}(z) \Big)\right|_{z=z_0} \neq 0$ ; alternatively, using Cauchy's theorem, the first integer $m_a$ such that $\oint_{\mathcal{C}_{z_0}}  \frac{\boldsymbol{M}(z)}{(z-z_0)^{m_a}} \mathrm{d}z = \boldsymbol{0}$.
The geometric multiplicity $m_g$ is the dimension of the kernel (nullspace), i.e. $m_g = \mathrm{dim}\left( \mathrm{Ker}\left(\boldsymbol{M}(z_0)\right)\right)$.
In general the algebraic multiplicity is greater than the geometric one: $m_a \geq m_g$.

$\boldsymbol{M}(z_0)$ is said to be \textit{semi-simple} if its geometric and algebraic multiplicities are equal, i.e. $m_a = m_g$ (c.f. theorem 2, p.120 in \cite{Elements_de_Mathematique_Algebre_Bourbaki_1959}).
Semi-simplicity can be established using the following result: $\boldsymbol{M}(z_0)$ is semi-simple if, and only if, for each nonzero $v \in \mathrm{Ker}\left( \boldsymbol{M}(z_0)\right)$, there exists $ w \in \mathrm{Ker}\left( \boldsymbol{M}(z_0)\right)$ such that:
\begin{equation}
v^\mathsf{T} \Bigg(\left. \frac{\mathrm{d}\boldsymbol{M}}{\mathrm{d}z} \right|_{z=z_0}  \Bigg)w \neq 0
\label{eq::semisimple condition}
\end{equation}

If an operator $\boldsymbol{M}(z_0)$ is semi-simple at a root $z_0$, then $z_0$ is a pole of Laurent order one for the inverse operator $\boldsymbol{M}^{-1}(z) \underset{\mathcal{W}(z_0)}{\sim} \frac{\boldsymbol{\widetilde{M}}}{z-z_0} $.
For Hermitian operators, the semi-simplicity property is guaranteed. However, resonances seldom correspond to Hermitian operators. In our case, the resonances correspond to the poles of the scattering matrix $\boldsymbol{U}(E)$, which is not self-adjoint but complex-symmetric $\boldsymbol{U}^\mathsf{T} = \boldsymbol{U}$ (c.f. equation (2.15) section VI.2.c p.287 in \cite{Lane_and_Thomas_1958}).
For complex-symmetric operators, semi-simplicity is not guaranteed in general, even when discarding the complex case of quasi-null vectors.

In the case of R-matrix theory, we were able to find cases where the geometric multiplicity of the scattering matrix does not match the algebraic one, thus R-matrix theory does not always yield semi-simple scattering matrices, and the Laurent development orders of the resonance poles can be higher.
For instance, we can devise examples of non-semi-simple inverse level matrices from definition (\ref{eq:inv_A expression}) by choosing resonance parameters such that the algebraic multiplicity is strictly greater than the geometric one.

However, one can also observe in these simple cases that the space of parameters for which semi-simplicity is broken is a hyper-plane of the space of R-matrix parameters. 
This gives credit to the traditional physics arguments that the probability of this occurring is quasi-null: R-matrix theory can yield scattering matrices with Laurent orders higher than one, but this is extremely unlikely; a mathematical approach of generic simplicity of resonances can be found in chapter 4 ``Black Box Scattering in $\mathbb{R}^n$" of \cite{Dyatlov}, in particular theorem 4.4 (Meromorphic continuation for black box Hamiltonians), theorem 4.5 (Spectrum of black box Hamiltonians), theorem 4.7 (Singular part of RV($\lambda$) for black box Hamiltonians), and theorem 4.39 (Generic simplicity of resonances for higher dimensional black box with potential perturbation).
In other terms, we assume semi-simplicity is almost always guaranteed through R-matrix parametrizations.

Henceforth, we use this argument to continue assuming the Kapur-Peierls matrix $\boldsymbol{R}_{L}$ is usually semi-simple, and thus the Laurent order of the radioactive poles $\big\{ \mathcal{E}_j \big\}$ in (\ref{eq::RL Mittag Leffler degenerate state}) is, in practice, one.

But let us be aware that in general scattering theory, the scattering operator may exhibit high-order poles \cite{Dyatlov, Guillope_1989_ENS, Owusu_2009}, and efforts are being made to have these ``exceptional points'' of second order arise in the specific case of nuclear interactions \cite{Exceptional_points_scattering, Exceptional_points_Hamiltonian_Michel}. The traditional R-matrix assumption where the poles of the scattering matrix are almost-always of Laurent-order one is unable to describe these physical phenomena.

\subsection{\label{subsec::Spurious poles cancellation for analytically continued scattering matrix} Scattering matrix poles are the Siegert-Humblet radioactive poles}

This section is dedicated to a remarkable property of the Siegert-Humblet radioactive poles $\left\{ \mathcal{E}_j \right\}$: in R-matrix theory, these are exactly the poles of the scattering matrix (theorem \ref{theo: Scattering matrix poles are the Siegert-Humblet radioactive poles}).

\begin{theorem}\label{theo: Scattering matrix poles are the Siegert-Humblet radioactive poles} \textsc{Scattering matrix poles are the Siegert-Humblet radioactive poles}. \\
In R-matrix theory, when the R-matrix operators (Kapur-Peierls $\boldsymbol{R}_L$ and incoming and outgoing wavefunctions $\boldsymbol{I}$ and $\boldsymbol{O}$) are analytically continued to complex energies $E\in\mathbb{C}$ such as to respect the Wronskian condition (\ref{eq:wronksian expression}), then  the poles of the scattering matrix $\boldsymbol{U}$ are exactly the poles of the Kapur-Peierls operator $\boldsymbol{R}_{L}$, i.e. the Siegert-Humblet radioactive poles $\big\{ \mathcal{E}_j \big\}$ from (\ref{eq:R_L radioactive problem}) and (\ref{eq:E_j pole def}).
These poles are almost always of Laurent-order of one.
\end{theorem}

Section \ref{subsubsec:: Semi-simple poles in R-matrix theory} gives the reasons to assume that the poles of the Kapur-Peierls matrix $\boldsymbol{R}_{L}$ are simple (i.e. or Laurent order one).
For the rest of this theorem, we here give two proofs: a first by \textit{reductio ad absurdum}, and a second constructive proof.

\begin{proof}
\textit{Reductio ad absurdum}: Since the radioactive poles $\mathcal{E}_j$ are not poles of the outgoing wavefunction, i.e. $\boldsymbol{O}^{-1} \boldsymbol{\rho}^{1/2}(\mathcal{E}_j) \neq \boldsymbol{0} $, expression (\ref{eq:U expression}) implies that all the poles $\mathcal{E}_j$ of the Kapur-Peierls $\boldsymbol{R}_L(E)$ operator are poles of the scattering matrix $\boldsymbol{U}(E)$. 
As first sight, expression (\ref{eq:U expression}) would suggest the roots $\left\{ \omega_n\right\}$ of the outgoing wavefunctions (i.e. all such that there exists a channel $c$ for which $O_c(\omega_n) = 0$) are also poles of the scattering matrix. 
However, when performing analytic continuation of R-matrix operators while conserving the Wronskian condition (\ref{eq:wronksian expression}), expression (\ref{eq:U expression}) is equivalent to expression eq. (1.5) of section VII.1 of \cite{Lane_and_Thomas_1958}, for which it is evident that the roots $\left\{ \omega_n\right\}$ of the outgoing wavefunction $O_c(\rho_c)$ are not poles of the scattering matrix $\boldsymbol{U}$ (that is because in both the Coulomb and the neutral particle case the outgoing wavefunctions $O_c(\rho_c)$ are confluent hypergeometric functions with simple roots $\left\{ \omega_n\right\}$ entailing that $O_c^{(1)}(\omega_n) \neq 0$).
Hence the poles of the scattering matrix $\boldsymbol{U}(E)$ must be exactly all the radioactive poles $\mathcal{E}_j$.
\end{proof}
Though this latter proof is correct, it does not explain how the roots $\left\{ \omega_n\right\}$ of the outgoing wavefunction $O_c(\rho_c)$ cancel out of the scattering matrix in expression (\ref{eq:U expression}). 
Yet it is important to understand this because expression (\ref{eq:U expression}) defines the potential cross section in standard nuclear data libraries, which taken as is should thus count the $\left\{ \omega_n\right\}$ as poles. We use this explicit cancellation of these spurious poles at the residues level to establish the windowed multipole representation in our follow-up article \cite{Ducru_WMP_Phys_Rev_C_2019}. 
Moreover, if one uses the Lane \& Thomas ``force closure'' definitions, then expression (\ref{eq:U expression}) and eq. (1.5) of section VII.1 of \cite{Lane_and_Thomas_1958} are no longer equivalent in the complex plane. In this case, not only is the scattering matrix $\boldsymbol{U}(E)$ no longer meromorphic, but it also diverges at the $\left\{ \omega_n\right\}$ outgoing wavefunction roots. 
Also, a constructive proof requires a closer look at the behavior of specific poles and residues, and gives us an opportunity to explain in detail different non-trivial assumptions usually made in nuclear physics about radioactive states and other states degeneracy. 
For all these reasons, we believe it of interest to here provide a second, constructive proof of theorem \ref{theo: Scattering matrix poles are the Siegert-Humblet radioactive poles}. 
It rests on the following lemma \ref{lem::Diagonal semi-simplicity Lemma}.

\begin{lem}\label{lem::Diagonal semi-simplicity Lemma}
\textsc{Diagonal Semi-Simplicity} --
If a diagonal matrix $\boldsymbol{D}^{-1}(z)$ is composed of elements with simple roots $\big\{\omega_n\big\}$, then its inverse is semi-simple, i.e. when a pole $\omega_n$ of a diagonal matrix $\boldsymbol{D}(z)$ has an algebraic multiplicity $M_n > 1$ the Laurent development order of the pole remains 1 while the associated residue matrix is of rank $M_n$, and can be expressed as:
\begin{equation}
\begin{IEEEeqnarraybox}[][c]{rcl}
\boldsymbol{D}(z) & \underset{\mathcal{W}(z=\omega_n)}{=} & \boldsymbol{D}_0 + \frac{\boldsymbol{D}_n}{z-\omega_n} \\
\boldsymbol{D}_n & \triangleq & \sum_{m=1}^{M_n}\frac{\boldsymbol{v_n}\boldsymbol{v_n}^\mathsf{T} }{\boldsymbol{v_n}^\mathsf{T} {\boldsymbol{D_0^{-1}}}^{(1)} \boldsymbol{v_n}}
\IEEEstrut\end{IEEEeqnarraybox}
\label{eq::higher algebraic multiplicity scattering poles}
\end{equation}
\end{lem}

\begin{proof}
Without loss of generality, a change of variables can be performed so as to set $\omega_n = 0$.
Let $\boldsymbol{D}(z) = \boldsymbol{\mathrm{diag}}\left( d_1(z), d_2(z), \hdots, d_1(z), d_j(z), d_n(z) \right) $ be a diagonal meromorphic complex-valued operator, which admits a pole at $z=0$.
$\boldsymbol{D^{-1}}(z) = \boldsymbol{\mathrm{diag}}\left( d_1^{-1}(z), d_2^{-1}(z), \hdots, d_1^{-1}(z), d_j^{-1}(z), d_n^{-1}(z) \right) $ is well known, and
$\mathrm{det}\left( \boldsymbol{D^{-1}} \right)(z=0) = d_1^{-1}(z)^2 \prod_{j \neq 1 }d_j^{-1}(z)$.
Let us assume only $d_1^{-1}(z=0) = 0$, with a simple root, so that $d_1(z) \underset{\mathcal{W}(z=0)}{=} d0_1 + \frac{R_1}{z}$.
Then $\mathrm{det}\left( \boldsymbol{D^{-1}}(z) \right)(z=0) = d_1^{-1}(z)^2 \prod_{j \neq 1 }d_j^{-1}(z)$ has a double root: the algebraic multiplicity is thus $2$.
However, it is immediate to notice that:
\begin{equation*}
\begin{IEEEeqnarraybox}[][c]{rcl}
\boldsymbol{D}(z) & = & \boldsymbol{\mathrm{diag}}\left( d_1(z), d_2(z), \hdots, d_1(z), d_j(z), d_n(z) \right) \\  & \underset{\mathcal{W}(z=0)}{=} & \boldsymbol{\mathrm{diag}}\left( d0_1, d_2(z), \hdots, d0_1(z), d_j(z), d_n(z) \right) \\
& & + \frac{1}{z} \boldsymbol{\mathrm{diag}}\left( R_1, 0, \hdots, R_1, 0, 0 \right)
\IEEEstrut\end{IEEEeqnarraybox}
\end{equation*}
This means the Laurent development order remains 1, albeit the algebraic multiplicity of the pole is 2 (or higher $M_n$). It can thus be written that: $\boldsymbol{D}(z) \underset{\mathcal{W}(z=0)}{=} \boldsymbol{D}_0 + \frac{\boldsymbol{D}_1}{z}$. 
When solving for the non-linear eigenproblem $\boldsymbol{D^{-1}}(z) \boldsymbol{v} = \boldsymbol{0}$, the kernel is no longer an eigenline, but instead spans $\left( \boldsymbol{v_1}, \boldsymbol{v_2}  \right)$, i.e. $\mathrm{Ker}\left(\boldsymbol{D^{-1}}_0\right) = \mathrm{Span}\left( \boldsymbol{v_1}, \boldsymbol{v_2}  \right)$, with $\boldsymbol{v_1} = a_1 \left[ 1, 0 , \hdots, 0, 0 ,0 \right]^\mathsf{T} $ and $\boldsymbol{v_2} = a_2 \left[ 0, 0 , \hdots, 1, 0 ,0 \right]^\mathsf{T} $.
Then, following Gohberg-Sigal's theory \cite{Gohberg_Sigal_1971}, the fundamental property:
\begin{equation*}
\boldsymbol{D^{-1}}\boldsymbol{D} = \Id{}
\end{equation*}
and the Laurent development around the pole:
\begin{equation*}
\begin{IEEEeqnarraybox}[][c]{rcl}
\boldsymbol{D^{-1}}(z) & \underset{\mathcal{W}(z=0)}{=} & \boldsymbol{D^{-1}}_0 + z {\boldsymbol{D^{-1}}}^{(1)}_0 + \mathcal{O}\left(z^2\right)
\IEEEstrut\end{IEEEeqnarraybox}
\end{equation*}
yield the relations:
\begin{equation*}
\begin{IEEEeqnarraybox}[][c]{rcl}
\boldsymbol{D^{-1}}_0  \boldsymbol{D}_0 + {\boldsymbol{D^{-1}}}^{(1)}_0 \boldsymbol{D}_1  & =  & \Id{} \\
\boldsymbol{D^{-1}}_0  \boldsymbol{D}_1 & =  & \boldsymbol{0}
\IEEEstrut\end{IEEEeqnarraybox}
\end{equation*}
Constructing $\boldsymbol{D}_1$ to satisfy the latter then entails \begin{equation*}
\begin{IEEEeqnarraybox}[][c]{rcl}
\boldsymbol{D}_1 & = & \frac{\boldsymbol{v_1}\boldsymbol{v_1}^\mathsf{T} }{\boldsymbol{v_1}^\mathsf{T} {\boldsymbol{D_0^{-1}}}^{(1)} \boldsymbol{v_1}} +  \frac{\boldsymbol{v_2}\boldsymbol{v_2}^\mathsf{T} }{\boldsymbol{v_2}^\mathsf{T} {\boldsymbol{D_0^{-1}}}^{(1)} \boldsymbol{v_2}}
\IEEEstrut\end{IEEEeqnarraybox}
\end{equation*}
where the transpose is used because the matrix is complex-symmetric.
This reasoning immediately generalizes to expression (\ref{eq::higher algebraic multiplicity scattering poles}).
\end{proof}

Let $\left\{\omega_n\right\}$ be all the roots of the outgoing-wave functions (i.e. the poles of inverse outgoing wave $\boldsymbol{O}^{-1}$), which we can find by solving the non-linear Eigenvalue problem:
\begin{equation}
\boldsymbol{O}(\omega_n) {\boldsymbol{{w_n}}}_m = \boldsymbol{0}
\end{equation}
Looking at (\ref{eq:U expression}) shows that the roots of the outgoing wave functions $\boldsymbol{O}$ could endow the scattering matrix with additional poles, through $\boldsymbol{O}^{-1}$, and that these poles could potentially have higher Laurent orders, since $\boldsymbol{O}^{-1}$ appears twice in expression (\ref{eq:U expression}). 
Yet, because $\boldsymbol{O}$ is diagonal with simple roots, lemma \ref{lem::Diagonal semi-simplicity Lemma} entails $\boldsymbol{O}^{-1}$ is semi-simple: the algebraic multiplicities are equal to the geometric multiplicities, and thus the poles $\left\{\omega_n\right\}$ all have Laurent order one.
Situations can arise where same-charge channels within the same total angular momentum $J^\pi$ will carry same angular momenta $\ell_c = \ell_{c'}$ and equal channel radii $a_c = a_{c'}$.
In that case, the geometric multiplicity $M_n$ of pole $\omega_n$ will be equal to the number of channels sharing the same functional outgoing waves $O_c = O_{c'}$.
Diagonal semi-simplicity lemma \ref{lem::Diagonal semi-simplicity Lemma} then establishes that the residue of $\boldsymbol{O}^{-1}$ associated to pole $\omega_n$ is now a diagonal rank-$M_n$ matrix, $\boldsymbol{D_n}$, expressed as:
\begin{equation}
\boldsymbol{D_n} = \sum_{m=1}^{M_n} \frac{{\boldsymbol{w_n}}_m {\boldsymbol{w_n}}_m^\mathsf{T}}{{\boldsymbol{w_n}}_m^\mathsf{T} \boldsymbol{O}^{(1)}(\omega_n) {\boldsymbol{w_n}}_m}
\label{eq::D_n residue of O-1}
\end{equation}
where $\boldsymbol{O}^{(1)}(\omega_n)$ designates the first derivative of $\boldsymbol{O}$, evaluated at the pole value $\omega_n$.
This establishes the existence of higher-rank residues associated to the inverse outgoing wave function $\boldsymbol{O}^{-1}$.
Notice that if the channel radii $\big\{a_c\big\}$ where chosen at random, these high-rank residues would almost never emerge (null probability).
However, since $a_c$ is chosen arbitrarily in the context of R-matrix theory, it is often the case that evaluators set $a_c$ to a fixed value for multiple different channels, and even across isotopes. This is because the scattering radius is determined early on by the evaluator (an not varied afterwards) based on the amount of potential scattering observed in the experimental data, which is very similar for isotopes of the same element.
Therefore, in practice these high-rank residues are not uncommon.
Our constructive proof now establishes how analytic continuation annuls these high-rank residues.

\begin{proof}
\textit{Constructive}: Consider the scattering matrix expression $\boldsymbol{U} =  \boldsymbol{O}^{-1} \left[\boldsymbol{I} + 2i \boldsymbol{\rho}^{1/2} \boldsymbol{R}_{L} \boldsymbol{O}^{-1} \boldsymbol{\rho}^{1/2}\right]$ from (\ref{eq:U expression}).
Result (\ref{eq::D_n residue of O-1}) entails that, in the vicinity of $\omega_n$ (root of the outgoing wave-function $\boldsymbol{O}$) the residue is locally given by:
\begin{equation}
\begin{IEEEeqnarraybox}[][c]{rcl}
\boldsymbol{U}(z) &  \underset{\mathcal{W}(E=\omega_n)}{=}  & \mkern-6mu \boldsymbol{U}_0(\omega_n) \mkern-3mu + \mkern-3mu \frac{ \boldsymbol{D}_n  \mkern-3mu \left[\boldsymbol{I} \mkern-5mu + 2i \boldsymbol{\rho}^{1/2} \boldsymbol{R}_{L} \boldsymbol{O}^{-1} \boldsymbol{\rho}^{1/2} \right]_{E=\omega_n} }{E-\omega_n}
\IEEEstrut\end{IEEEeqnarraybox}
\label{eq:: U at omega_n}
\end{equation}
We now notice that evaluating the Kapur-Peierls $\boldsymbol{R}_{L}$ operator (\ref{eq: def Kapur-Peierls operator}) at the pole value $\omega_n$ yields the following equality:
\begin{equation}
\begin{IEEEeqnarraybox}[][c]{rcl}
\boldsymbol{R}_{L} \boldsymbol{O}^{-1} (\omega_n) {\boldsymbol{w_n}}_m  & \ = \ & - \left[ \boldsymbol{\rho} \boldsymbol{O}^{(1)} \right]^{-1}(\omega_n) {\boldsymbol{w_n}}_m
\IEEEstrut\end{IEEEeqnarraybox}
\label{eq::R_L O at outgoing scattering poles}
\end{equation}
Plugging (\ref{eq::R_L O at outgoing scattering poles}) into the residue of (\ref{eq:: U at omega_n}), and using the fact that (\ref{eq::D_n residue of O-1}) guarantees $\boldsymbol{D_n} $ is a linear combination of $ {\boldsymbol{w_n}}_m {\boldsymbol{w_n}}_m^\mathsf{T} $, we then have the following equality on the residues at poles $\omega_n$:
\begin{equation}
\begin{IEEEeqnarraybox}[][c]{rcl}
\boldsymbol{D}_n \mkern-3mu \left[\boldsymbol{I} \mkern-5mu + 2i \boldsymbol{\rho}^{1/2} \boldsymbol{R}_{L} \boldsymbol{O}^{-1} \boldsymbol{\rho}^{1/2} \right]_{E=\omega_n} \mkern-20mu = \boldsymbol{D}_n \mkern-3mu \left[ \boldsymbol{I} - 2i {\boldsymbol{O}^{(1)}}^{-1} \right]_{E=\omega_n}
\IEEEstrut\end{IEEEeqnarraybox}
\label{eq::U residues at outgoing scattering poles}
\end{equation}
The rightmost term is diagonal and independent from the resonance parameters.
Since the Wronskian matrix $\boldsymbol{w}$ of the external region interaction (for Coulomb or free particles) is constant (\ref{eq:wronksian expression}), $\boldsymbol{w} = \boldsymbol{O}^{(1)}\boldsymbol{I} - \boldsymbol{I}^{(1)}\boldsymbol{O} = 2i \Id{}$, evaluating at outgoing wave-function root $\omega_n$, one finds $2i \Id{} = \boldsymbol{O}^{(1)}\boldsymbol{I}(\omega_n)$.
Plugging this result into (\ref{eq::U residues at outgoing scattering poles}) annuls the corresponding residue from the scattering matrix, i.e.:
\begin{equation}
\begin{IEEEeqnarraybox}[][c]{c}
\boldsymbol{D}_n \left[\boldsymbol{I} + 2i \boldsymbol{\rho}^{1/2} \boldsymbol{R}_{L} \boldsymbol{O}^{-1} \boldsymbol{\rho}^{1/2} \right]_{E=\omega_n}  = \boldsymbol{0}
\IEEEstrut\end{IEEEeqnarraybox}
\end{equation}
Thus, if the Wronskian condition (\ref{eq:wronksian expression}) is respected, the $\big\{\omega_n\big\}$ poles cancel out of the scattering matrix $\boldsymbol{U}$
\end{proof}

Importantly, both the Lane \& Thomas force-closing of sub-threshold channels \ref{subsec::Forcing sub-threshold elements to zero} or the analytic continuation \ref{subsec::Analytic continuation of the scattering matrix} will yield the same cross section values for real energies above thresholds.
However, theorem \ref{theo: Scattering matrix poles are the Siegert-Humblet radioactive poles} demonstrates that the choice of analytic continuation in equation (\ref{eq:U expression}), respecting the Wronskian condition (\ref{eq:wronksian expression}), leads to the cancellation from the scattering matrix $\boldsymbol{U}$ of the $\big\{\omega_n\big\}$ spurious poles, which have nothing to do with the resonant states of the scattering system.
This cancellation is thus physically accurate, and would not take place had the choice of $\boldsymbol{\mathfrak{P}} = \boldsymbol{P}$ been made in equation (\ref{eq:U expression L&T}) with the Lane \& Thomas ``force closure'' definition $\boldsymbol{P} = \Im\left[ \boldsymbol{L}(z) \right] \; \in \mathbb{R}$ (c.f. lemma 1 section III.B of \cite{Ducru_shadow_Brune_Poles_2020}), under which the scattering matrix diverges at $\big\{\omega_n\big\}$. Conversely, analytically continuing the penetration function as $ \boldsymbol{P}(z) \triangleq \frac{1}{2\mathrm{i}}\left( \boldsymbol{L}(z) - \left[\boldsymbol{L}(z^*)\right]^* \right)  \; \in \mathbb{C}$ (c.f. lemma 2 section III.B of \cite{Ducru_shadow_Brune_Poles_2020}) will guarantee the cancellation of the $\big\{\omega_n\big\}$ poles from the scattering matrix $\boldsymbol{U}$ when using (\ref{eq:U expression L&T}).
Notice this is almost the definition (\ref{eq:Delta_L def}) of $\boldsymbol{ \Delta L}(\rho)$ we hereafter use in the proof of the generalized unitarity.
Then, to force-close sub-threshold channels, one could set the Wronskian to zero, as proposed by Lane \& Thomas in the paragraph between equation (2.1) and (2.2) of section VII.1. p.289. This shifts the problem to how to maintain the Wronskian condition (\ref{eq:wronksian expression}) while setting the Wronskian to zero below thresholds. 
Alternatively, we here argue in section \ref{subsec::Closure of sub-threshold cross sections through analytic continuation} that this might not be necessary, as analytic continuation can naturally close sub-threshold channels.

\subsection{\label{subsec::Generalized unitarity} Generalized unitarity for analytically continued scattering matrix}

One of the authors, G. Hale, proved a somewhat more esoteric argument in favor of analytic continuation of the scattering matrix, showing it satisfies generalized unitarity. 

Eden \& Taylor established a generalized unitarity condition, eq. (2.16) in \cite{Eden_and_Taylor}, which extents the one described by Lane \& Thomas, eq. (2.13), VI.2.c. p.287, in that the subset of open channels is unitary (thus conserving probability), but the scattering matrix can still be continued to sub-threshold channels and be non-zero, that is the full scattering matrix of open and closed channels is not unitary but satisfies the generalized unitarity condition. This is also consistent with approaches other than R-matrix to modeling nuclear interactions (c.f. commentary above eq. (3) p.4 in \cite{Exceptional_points_Hamiltonian_Michel}, \cite{complex_energy_above_inoization_threshold_1969}, or\cite{S-matrix-complex_resonance_parameters_Csoto_1997}).

The premises of the problem lies again in the multi-sheeted Riemann surface spawning from mapping (\ref{eq:rho_c(E) mapping}): when considering the scattering matrix $\boldsymbol{U}(E)$ at a given energy $E$, there are multiple possibilities for the choice of wavenumber $k_c$ at each channel. 
Following Eden \& Taylor eq. (2.14a) and eq. (2.14b) \cite{Eden_and_Taylor}, we consider the case of momenta being continued along the following paths in the multi-sheeted Riemann surface: one subset of channels $c$, denoted by $\widehat{\mathfrak{C}}$, is continued as $k_{c \in \widehat{\mathfrak{C}}} \to k_{c \in \widehat{\mathfrak{C}}}^*  $, while all the others are continued as $k_{c \not\in \widehat{\mathfrak{C}}} \to - k_{c \not\in \widehat{\mathfrak{C}}}^*  $, and we collectively denote this continuation $\boldsymbol{k} \to \boldsymbol{\widetilde{k}}$:
\begin{equation}
\boldsymbol{k} \to \boldsymbol{\widetilde{k}} : \left\{ \begin{array}{rcl}
\forall c \in \widehat{\mathfrak{C}}\; , \quad k_c & \to & k_c^* \\
\forall c \not\in \widehat{\mathfrak{C}}\; , \quad k_c & \to & -k_c^* \\
\end{array}\right. 
\label{eq:: Eden and Taylor continuation}
\end{equation}
We then seek to reproduce the generalized unitarity property eq. (2.16) of \cite{Eden_and_Taylor}, which states that the submatrix $\widehat{\boldsymbol{U}}$ composed of the channels ${c \in \widehat{\mathfrak{C}}}$, verifies the generalized unitarity condition:
\begin{equation}
\widehat{\boldsymbol{U}}(\boldsymbol{k}) \Big[ \widehat{\boldsymbol{U}}(\boldsymbol{\widetilde{k}}) \Big]^\dagger = \Id{}
\label{eq:: Eden and Taylor Generalized unitarity}
\end{equation}
We now show that analytically continuing the R-matrix expression (\ref{eq:U expression}) ensures the scattering matrix respects Eden \& Taylor generalized unitarity condition.

\begin{theorem}\label{theo::satisfied generalized unitarity condition} \textsc{Analytic continuation of the R-matrix expression for the scattering matrix ensures generalized unitarity}. \\
By performing the analytic continuation of the R-matrix expression (\ref{eq:U expression}), the scattering matrix $\boldsymbol{U}$ satisfies Eden \& Taylor's generalized unitarity condition (\ref{eq:: Eden and Taylor Generalized unitarity}).
\end{theorem}

\begin{proof}
The proof is based on the conjugacy relations of the outgoing and incoming wavefunctions eq. (2.12), VI.2.c. in \cite{Lane_and_Thomas_1958}, whereby, for any channel $c$:
\begin{equation}
\begin{IEEEeqnarraybox}[][c]{rclcrcl}
\big[O_c(k_c^*)\big]^* & \ = \ & I_c(k_c) \quad  &,& \quad \big[I_c(k_c^*)\big]^* & \ = \ & O_c(k_c) \\
O_c(-k_c) & \ = \ & I_c(k_c) \quad &,& \quad I_c(-k_c) & \ = \ & O_c(k_c) \\ -O_c^{(1)}(-k_c) & \ = \ & I_c^{(1)}(k_c) \quad &,& \quad -I_c^{(1)}(-k_c) & \ = \ & O_c^{(1)}(k_c) \\
\IEEEstrut\end{IEEEeqnarraybox}
\label{eq:Conjugacy for O and I}
\end{equation}
where the third line was obtained by taking the derivative of the second. 
Conjugacy relations (\ref{eq:Conjugacy for O and I}) entail the following relations on the outgoing-wave reduced logarithmic derivative $\boldsymbol{L}$:
\begin{equation}
\begin{IEEEeqnarraybox}[][c]{rclcrcl}
\Big[L_c(k_c^*)\Big]^* & \ = \ & L_c(-k_c) \quad  &,& \quad \Big[L_c(-k_c^*)\Big]^* & \ = \ & L_c(k_c)
\IEEEstrut\end{IEEEeqnarraybox}
\label{eq:L_c congugacy}
\end{equation}
We also notice that the Wronskian condition (\ref{eq:wronksian expression}) is equivalent to:
\begin{equation}
\begin{IEEEeqnarraybox}[][c]{rcl}
\frac{2\mathrm{i}\rho_c}{O_c I_c} & \ = \ & \rho_c\left[\frac{O_c^{(1)}}{O_c} - \frac{I_c^{(1)}}{I_c}\right]
\IEEEstrut\end{IEEEeqnarraybox}
\label{eq:wronksian expression for generalized unitarity}
\end{equation}
Recognizing the definition (\ref{eq: L operator}) of $\boldsymbol{L}$, and using conjugacy relations (\ref{eq:L_c congugacy}), this  Wronskian condition (\ref{eq:wronksian expression for generalized unitarity}) can be expressed as a difference of the reduced logarithmic $L_c$ derivatives:
\begin{equation}
\begin{IEEEeqnarraybox}[][c]{rcl}
 \Delta L_c (k_c)&  \ \triangleq \ & L_c(k_c) - L_c(-k_c) \ = \ \frac{2\mathrm{i}\rho_c}{O_c I_c}(k_c)
\IEEEstrut\end{IEEEeqnarraybox}
\label{eq:Delta_L def}
\end{equation}
Defining the diagonal matrix $\boldsymbol{\Delta L} \triangleq \boldsymbol{\mathrm{diag}}\Big(\Delta L_c (k_c)\Big)$, we can then re-write, similarly to (\ref{eq:U expression L&T}), the R-matrix expression (\ref{eq:U expression}) of the scattering matrix $\boldsymbol{U}$ as a function of $ \Delta L_c (k_c)$, so that:
\begin{equation}
\begin{IEEEeqnarraybox}[][c]{rcl}
 \boldsymbol{U} &  \ = \ & \boldsymbol{O}^{-1}\left[ \Id{} + \left[ \boldsymbol{\rho}^{1/2}  \boldsymbol{R}_{L} \boldsymbol{\rho}^{-1/2} \right] \boldsymbol{\Delta L}  \right] \boldsymbol{I}\\
 & \ = \ &  \boldsymbol{I} \left[ \Id{} + \boldsymbol{\Delta L} \left[ \boldsymbol{\rho}^{-1/2}  \boldsymbol{R}_{L} \boldsymbol{\rho}^{1/2} \right]   \right]  \boldsymbol{O}^{-1}
\IEEEstrut\end{IEEEeqnarraybox}
\label{eq: U as a function of Delta_L}
\end{equation}
Notice again how this expression is closely related to the analytic continuation of expression (\ref{eq:U expression L&T}).

Coming back to the Eden \& Taylor continuation (\ref{eq:: Eden and Taylor continuation}), let us now establish a relation between the Kapur-Peierls operator $\boldsymbol{R}_{L}$ and $\boldsymbol{\Delta L} $. 
From the definition (\ref{eq: def Kapur-Peierls operator}) of the Kapur-Peierls operator $\boldsymbol{R}_{L}$, recalling that under Eden \& Taylor continuations (\ref{eq:: Eden and Taylor continuation}) the energy $E$ from mapping (\ref{eq:rho_c(E) mapping}) remains unaltered, and given that the boundary-condition $B_c$ in the $\boldsymbol{L^0}$ matrix function is real and thus the R-matrix parameters (\ref{eq:R expression}) are too, it follows that:
\begin{equation}
\begin{IEEEeqnarraybox}[][c]{rcl}
  \left[\boldsymbol{R}_{L}^{-1}(\boldsymbol{\widetilde{k}} ) \right]^*&  -\boldsymbol{R}_{L}^{-1}(\boldsymbol{k})  \ = \ & \left( \begin{array}{cc}
\widehat{\boldsymbol{\Delta L}}(\boldsymbol{k}) & 0 \\
0 & 0 \\
\end{array}\right) 
\IEEEstrut\end{IEEEeqnarraybox}
\label{eq:Delta_L and R_L}
\end{equation}
where we have used the $\boldsymbol{L}$ conjugacy relations (\ref{eq:L_c congugacy}) to establish that all channels $c \not\in \widehat{\mathfrak{C}}$ cancel out, and the rest yield $\Delta L_{c \in \widehat{\mathfrak{C}}}(k_c)$. The $\widehat{\boldsymbol{\Delta L}}$ thus designates the sub-matrix composed of all the channels $c \in \widehat{\mathfrak{C}}$.
Multiplying both left and right, and considering the sub-matrices on the channels $c \in \widehat{\mathfrak{C}}$ thus yields:
\begin{equation}
\begin{IEEEeqnarraybox}[][c]{rcl}
    \boldsymbol{\widehat{R}}_L(\boldsymbol{k}) - \left[\boldsymbol{\widehat{R}}_L(\boldsymbol{\widetilde{k}} ) \right]^*& \ = \ & \boldsymbol{\widehat{R}}_L(\boldsymbol{k}) \widehat{\boldsymbol{\Delta L}}(\boldsymbol{k}) \left[\boldsymbol{\widehat{R}}_L(\boldsymbol{\widetilde{k}} ) \right]^*
\IEEEstrut\end{IEEEeqnarraybox}
\label{eq:Delta_L and R_L submatrix}
\end{equation}
This relation is what guarantees the scattering matrix $\boldsymbol{U}$ satisfies generalized unitarity condition (\ref{eq:: Eden and Taylor Generalized unitarity}).
Indeed, let us develop the left-hand side of (\ref{eq:: Eden and Taylor Generalized unitarity}), using expressions (\ref{eq: U as a function of Delta_L}) on the sub-matrices of the channels $c \in \widehat{\mathfrak{C}}$:
\begin{equation}
\begin{IEEEeqnarraybox}[][c]{l}
    \widehat{\boldsymbol{U}}(\boldsymbol{k}) \Big[ \widehat{\boldsymbol{U}}(\boldsymbol{\widetilde{k}}) \Big]^\dagger \ = \  \\ \boldsymbol{\widehat{O}}^{-1}(\boldsymbol{k})\left[ \Id{} + \widehat{\left[ \boldsymbol{\rho}^{\frac{1}{2}}  \boldsymbol{R}_{L} \boldsymbol{\rho}^{-\frac{1}{2}} \right]}(\boldsymbol{k}) \widehat{\boldsymbol{\Delta L}}(\boldsymbol{k})  \right] \boldsymbol{\widehat{I}}(\boldsymbol{k})  \\
     \  \   \times  \left[ \boldsymbol{\widehat{I}}(\boldsymbol{\widetilde{k}}) \left[ \Id{} + \widehat{\boldsymbol{\Delta L}}(\boldsymbol{\widetilde{k}}) \widehat{\left[ \boldsymbol{\rho}^{-\frac{1}{2}}  \boldsymbol{R}_{L} \boldsymbol{\rho}^{\frac{1}{2}} \right]}  (\widetilde{\boldsymbol{k}}) \right]  \boldsymbol{\widehat{O}}^{-1}(\boldsymbol{\widetilde{k}}) \right]^\dagger \\
     \ = \ \boldsymbol{\widehat{O}}^{-1}(\boldsymbol{k})\left[ \Id{} + \widehat{\left[ \boldsymbol{\rho}^{\frac{1}{2}}  \boldsymbol{R}_{L} \boldsymbol{\rho}^{-\frac{1}{2}} \right]}(\boldsymbol{k}) \widehat{\boldsymbol{\Delta L}}(\boldsymbol{k})  \right] \boldsymbol{\widehat{I}}(\boldsymbol{k}) \times \\
     \  \   \left[ \boldsymbol{\widehat{O}}^{-1}(\boldsymbol{\widehat{k}^*})\right]^{\mkern-2mu*}  \mkern-6mu \left[ \Id{} \mkern-3mu + \mkern-6mu \widehat{\left[ \boldsymbol{\rho}^{-\frac{1}{2}}  \boldsymbol{R}_{L} \boldsymbol{\rho}^{\frac{1}{2}} \right]}  (\boldsymbol{\widehat{k}^*}) \left[\widehat{\boldsymbol{\Delta L}}(\boldsymbol{\widehat{k}^*})\right]^* \right] \mkern-6mu \left[\boldsymbol{\widehat{I}}(\boldsymbol{\widehat{k}^*}) \right]^{\mkern-2mu*} \\
\IEEEstrut\end{IEEEeqnarraybox}
\label{eq:U developments for generalized unitarity}
\end{equation}
Noticing that conjugacy relation (\ref{eq:L_c congugacy}) entail the following $\boldsymbol{\Delta L}$ symmetry from definition (\ref{eq:Delta_L def}), $\left[\widehat{\boldsymbol{\Delta L}}(\boldsymbol{\widehat{k}^*})\right]^* = -\widehat{\boldsymbol{\Delta L}}(\boldsymbol{k}) $, and making use of the conjugacy relations for the wave functions (\ref{eq:Conjugacy for O and I}), we can further simplify (\ref{eq:U developments for generalized unitarity}) to:
\begin{equation}
\begin{IEEEeqnarraybox}[][c]{l}
 \widehat{\boldsymbol{U}}(\boldsymbol{k}) \Big[ \widehat{\boldsymbol{U}}(\boldsymbol{\widetilde{k}}) \Big]^\dagger  \ = \ \Id{}  \ + \  \boldsymbol{\widehat{O}}^{-1}(\boldsymbol{k})   \widehat{\left[ \boldsymbol{\rho}^{\frac{1}{2}}  \boldsymbol{R}_{L} \boldsymbol{\rho}^{-\frac{1}{2}} \right]}(\boldsymbol{k}) \times \Bigg[  \\  \left[\left[ \widehat{ \boldsymbol{\rho}^{\frac{1}{2}}  \boldsymbol{R}_{L} \boldsymbol{\rho}^{-\frac{1}{2}} }\right]^{-1}(\boldsymbol{k^*})\right]^\dagger   - \left[ \widehat{ \boldsymbol{\rho}^{\frac{1}{2}}  \boldsymbol{R}_{L} \boldsymbol{\rho}^{-\frac{1}{2}} }\right]^{-1}(\boldsymbol{k})   - \widehat{\boldsymbol{\Delta L}}(\boldsymbol{k})  \\ \quad \quad \quad \quad \Bigg] 
   \times  \left[ \widehat{\left[ \boldsymbol{\rho}^{-\frac{1}{2}}  \boldsymbol{R}_{L} \boldsymbol{\rho}^{\frac{1}{2}} \right]}  (\boldsymbol{\widehat{k}^*}) \right]^\dagger \widehat{\boldsymbol{\Delta L}}(\boldsymbol{k}) \boldsymbol{\widehat{O}}(\boldsymbol{k})
\IEEEstrut\end{IEEEeqnarraybox}
\label{eq:U developed for generalized unitarity}
\end{equation}
In the middle, we recognize property (\ref{eq:Delta_L and R_L}), where the $\boldsymbol{\rho}^{\pm 1/2}$ cancel out by commuting with the diagonal matrix.
Property (\ref{eq:Delta_L and R_L}) thus annuls all non-identity terms, leaving Eden \& Taylor's generalized unitarity condition (\ref{eq:: Eden and Taylor Generalized unitarity}) satisfied.
\end{proof}

Let us also note that the proof required real boundary conditions $B_c \in \mathbb{R}$.
Thus, in R-matrix parametrization (\ref{eq:U expression}), real boundary conditions $B_c \in \mathbb{R}$ are necessary for the scattering matrix $\boldsymbol{U}$ to be unitarity (and by extension generalized unitary).

Theorem \ref{theo::satisfied generalized unitarity condition} beholds a strong argument in favor of performing analytic continuation of the R-matrix operators as the physically correct way of prolonging the scattering matrix to complex wavenumbers $k_c\in\mathbb{C}$.

\subsection{\label{subsec::Closure of sub-threshold cross sections through analytic continuation} Closure of sub-threshold cross sections through analytic continuation}

We finish this article with the key question of how to close sub-threshold channels.
Analytically continuing the scattering matrix below thresholds entails it cannot be identically zero there, since this would entail it is the null function on the entire sheet of the maniforld (unicity of analytic continuation).
However, we here show that for massive particles subject to $\rho(E)$ mappings (2) or (4) section II.A of \cite{Ducru_shadow_Brune_Poles_2020}, adequate definitions and careful consideration will both make the transmission matrix evanescent sub-threshold (in a classical case of quantum tunnelling), and annul the sub-threshold cross-section --- the physically measurable quantity.

The equations linking the scattering matrix $\boldsymbol{U}$ to the cross section --- equations (1.9), (1.10) and (2.4) section VIII.1. of \cite{Lane_and_Thomas_1958} pp.291-293 --- were only derived for real positive wavenumbers. 
Yet, when performing analytic continuation of them to sub-threshold energies, the quantum tunneling effect will naturally make the transmission matrix infinitesimal on the physical sheet of mapping (\ref{eq:rho_c(E) mapping}).
Indeed, the \textit{transmission matrix}, $\boldsymbol{T}$, is defined in \cite{Lane_and_Thomas_1958} after eq. (2.3), VIII.2. p.292, as:
\begin{equation}
\begin{IEEEeqnarraybox}[][c]{rcl}
T_{cc'} & \ \triangleq \ &  \delta_{cc'}\mathrm{e}^{2\mathrm{i}\omega_c} - U_{cc'} 
\IEEEstrut\end{IEEEeqnarraybox}
\label{eq:Lane and Thomas Transmission matrix definition}
\end{equation}
where $\omega_c$ is defined by Lane \& Thomas in eq.(2.13c) III.2.b. p.269, and used in eq.(4.5a) III.4.a. p.271 in \cite{Lane_and_Thomas_1958}, and is the difference $ \omega_c = \sigma_{\ell_c}(\eta_c) - \sigma_{0}(\eta_c)$,  where the \textit{Coulomb phase shift}, $\sigma_{\ell_c}(\eta_c) $ is defined by Ian Thompson in eq.(33.2.10) of \cite{NIST_DLMF}. 
Defining the diagonal matrix $\boldsymbol{\omega} \triangleq \boldsymbol{\mathrm{diag}}\big( \omega_c \big)$, and using the R-matrix expression (\ref{eq:U expression}) for the scattering matrix, the Lane \& Thomas transmission matrix (\ref{eq:Lane and Thomas Transmission matrix definition}) can be expressed with R-matrix parameters as:
\begin{equation}
\begin{IEEEeqnarraybox}[][c]{rcl}
\boldsymbol{T}_{\text{L\&T}} \ \triangleq \ -2\mathrm{i} \boldsymbol{O}^{-1}\left[\underbrace{\left(\frac{\boldsymbol{I} - \boldsymbol{O}\boldsymbol{\mathrm{e}}^{2\mathrm{i}\boldsymbol{\omega}}}{2\mathrm{i}}\right)}_{\boldsymbol{\Theta}} + \boldsymbol{\rho}^{1/2} \boldsymbol{R}_{L}  \boldsymbol{O}^{-1} \boldsymbol{\rho}^{1/2} \right]
\IEEEstrut\end{IEEEeqnarraybox}
\label{eq:Transmission matrix Lane and Thomas matrix expression}
\end{equation}
The angle-integrated partial cross sections $\sigma_{cc'}(E)$ can then be expressed as eq.(3.2d) VIII.3. p.293 of \cite{Lane_and_Thomas_1958}:
\begin{equation}
\begin{IEEEeqnarraybox}[][c]{rcl}
\sigma_{cc'}(E) & \ = \ & \pi g_{J^\pi_c} \left| \frac{ T_{\text{L\&T}}^{cc'}(E)}{k_c(E)}\right|^2
\IEEEstrut\end{IEEEeqnarraybox}
\label{eq:partial sigma_cc'Lane and Thomas}
\end{equation}
where $ g_{J^\pi_c} \ \triangleq \ \frac{2 J + 1 }{\left(2 I_1 + 1 \right)\left(2 I_2 + 1 \right) }$ is the \textit{spin statistical factor} defined eq.(3.2c) VIII.3. p.293.
Plugging-in the transmission matrix R-matrix parametrization (\ref{eq:Transmission matrix Lane and Thomas matrix expression}) into  cross-section expression (\ref{eq:partial sigma_cc'Lane and Thomas}) then yields:
\cite{Lane_and_Thomas_1958}:
\begin{equation}
\begin{IEEEeqnarraybox}[][c]{rcl}
\sigma_{cc'} & \ = \ & 4\pi g_{J^\pi_c} \left|\frac{1}{O_c k_c} \right|^2  \left| \boldsymbol{\Theta} + \boldsymbol{\rho}^{1/2} \boldsymbol{R}_{L}  \boldsymbol{O}^{-1} \boldsymbol{\rho}^{1/2} \right|_{cc'}^2
\IEEEstrut\end{IEEEeqnarraybox}
\label{eq:partial sigma_cc'Lane and Thomas R-matrix}
\end{equation}

An alternative, more numerically stable, way of computing the cross section is used at Los Alamos National Laboratory, where one of the authors, G. Hale, introduced the following rotated transmission matrix, defined as:
\begin{equation}
\begin{IEEEeqnarraybox}[][c]{rcl}
\boldsymbol{T}_{\text{H}} & \ \triangleq \ & - \frac{ \boldsymbol{\mathrm{e}}^{-\mathrm{i}\boldsymbol{\omega}} \boldsymbol{T}_{\text{L\&T}} \boldsymbol{\mathrm{e}}^{-\mathrm{i}\boldsymbol{\omega}}}{2\mathrm{i}}
\IEEEstrut\end{IEEEeqnarraybox}
\label{eq:Hale Transmission matrix}
\end{equation}
and whose R-matrix parametrization is thus
\begin{equation}
\begin{IEEEeqnarraybox}[][c]{rcl}
\boldsymbol{T}_{\text{H}}  & \ = \ & \boldsymbol{H_+}^{-1}\left[ \boldsymbol{\rho}^{1/2} \boldsymbol{R}_{L} \boldsymbol{\rho}^{1/2} \boldsymbol{H_+}^{-1}  - \underbrace{\left(\frac{\boldsymbol{H_+} - \boldsymbol{H_-}}{2\mathrm{i}}\right)}_{\boldsymbol{F}} \right]
\IEEEstrut\end{IEEEeqnarraybox}
\label{eq:Hale Transmission matrix expression}
\end{equation}
where $\boldsymbol{H_\pm}$ are defined as in eq.(2.13a)-(2.13b) III.2.b p.269 \cite{Lane_and_Thomas_1958}:
\begin{equation}
\begin{IEEEeqnarraybox}[][c]{rcl}
{H_{+}}_c & \ = \ &  O_c \mathrm{e}^{\mathrm{i} \omega_c} = G_c + \mathrm{i} F_c \\
{H_{-}}_c & \ = \ &  I_c \mathrm{e}^{-\mathrm{i} \omega_c} = G_c - \mathrm{i} F_c
\IEEEstrut\end{IEEEeqnarraybox}
\label{eq:def H_pm I and O}
\end{equation}
and for which we refer to Ian J. Thompson's Chapter 33, eq.(33.2.11) in \cite{NIST_DLMF}, or Abramowitz \& Stegun chapter 14, p.537 \cite{Abramowitz_and_Stegun}.
The partial cross section is then directly related to the $\boldsymbol{T}_{\text{H}}$ rotated transmission matrix (\ref{eq:Hale Transmission matrix}) as:
\begin{equation}
\begin{IEEEeqnarraybox}[][c]{rcl}
\sigma_{cc'}(E) & \ = \ & 4\pi g_{J^\pi_c} \left|\frac{T_H^{cc'}(E)}{k_c(E)} \right|^2 
\IEEEstrut\end{IEEEeqnarraybox}
\label{eq:partial sigma_cc' Hale T-matrix}
\end{equation}

\begin{theorem}\label{theo::evanescence of sub-threshold transmission matrix} \textsc{Evanescence of sub-threshold transmission matrix}. \\
For massive particles, analytic continuation of R-matrix parametrization (\ref{eq:U expression}) makes the sub-threshold transmission matrix $\boldsymbol{T}$, defined as (\ref{eq:Transmission matrix Lane and Thomas matrix expression}), evanescent on the physical sheets of wavenumber-energy $\rho(E)$ mappings (2) or (4) section II.A of \cite{Ducru_shadow_Brune_Poles_2020}.
In turn, this quantum tunnelling entails the partial cross sections $\sigma_{cc'}(E)$ become infinitesimal below threshold.
\end{theorem}

\begin{proof}
The proof is based on noticing that both transmission matrix expressions (\ref{eq:Transmission matrix Lane and Thomas matrix expression}) and (\ref{eq:Hale Transmission matrix}) entail their modulus square is proportional to:
\begin{equation}
\begin{IEEEeqnarraybox}[][c]{rcl}
|\boldsymbol{T}_{cc'}|^2(E) & \ \propto \ &  \left|\frac{1}{H_+(E)} \right|^2 
\IEEEstrut\end{IEEEeqnarraybox}
\label{eq:T_cc' proportional to H}
\end{equation}
This is because $\boldsymbol{R}_{L}  \boldsymbol{O}^{-1} = \left[ \boldsymbol{O}\left[ \boldsymbol{R}^{-1} - \boldsymbol{B}\right] - \boldsymbol{\rho}\boldsymbol{O}^{(1)}\right]^{-1}$, which does not diverge below threshold. 
Asymptotic expressions for the behavior of $H_+(\rho)$ then yield, for small $\rho$ values:
\begin{equation}
\begin{IEEEeqnarraybox}[][c]{rcl}
H_{+}(\rho) & \ \underset{\rho \to 0}{\sim} \ &  \frac{\rho^{-\ell}}{(2\ell + 1)C_\ell(\eta)} - \mathrm{i} C_\ell(\eta)\rho^{\ell+1}
\IEEEstrut\end{IEEEeqnarraybox}
\label{eq:H+ near zero}
\end{equation}
and asymptotic large-$\rho$ behavior: 
\begin{equation}
\begin{IEEEeqnarraybox}[][c]{rcl}
H_{+}(\rho) & \ \underset{\rho \to \infty} {\sim}\ & \mathrm{e}^{\mathrm{i} (\rho - \eta \ln(2\rho) - \frac{1}{2}\ell\pi + \sigma_\ell(\eta) ) }
\IEEEstrut\end{IEEEeqnarraybox}
\label{eq:H+ near infinity}
\end{equation}
Above the threshold, $\rho \in \mathbb{R}$ is real and thus equation (\ref{eq:H+ near infinity}) shows how $| H_{+}(\rho) | \ \underset{\rho \to \infty} {\longrightarrow} 1  $. In other terms, the $| H_{+}(\rho) |$ term cancels out of the cross section expressions (\ref{eq:partial sigma_cc'Lane and Thomas R-matrix}) and (\ref{eq:partial sigma_cc' Hale T-matrix}) for open-channels above threshold.

Yet, in both wavenumber-energy $\rho(E)$ mappings (2) or (4) section II.A of \cite{Ducru_shadow_Brune_Poles_2020}, the sub-threshold dimensionless wavenumber is purely imaginary: $\rho \in \mathrm{i}\mathbb{R}$. 
Since asymptotic form (\ref{eq:H+ near infinity}) is dominated in modulus by: $\left|H_{+}(\rho)\right|  \underset{\rho \to \infty} {\sim} \left| \mathrm{e}^{\mathrm{i} \rho }\right|$.
Depending on which sheet $\rho$ is continued sub-threshold, we can have $\rho = \pm \mathrm{i} x$, with $x\in \mathbb{R}$. Thus, on the non-physical sheet $\big\{ E, \hdots, -_c , \hdots \big\}$ for the given channel $c$ of $\rho_c$, the transmission matrix (\ref{eq:T_cc' proportional to H}) experiences exponential decay of $1/\left|H_{+}(\rho)\right|$ leading to the evanescence of the cross section (\ref{eq:partial sigma_cc'Lane and Thomas}), or (\ref{eq:partial sigma_cc' Hale T-matrix}). 
In effect, this means that the $|O_c(\rho_c)|$ term in (\ref{eq:partial sigma_cc'Lane and Thomas R-matrix}) asymptotically acts like a Heaviside function, being unity for open channels, but closing the channels below threshold. 
Since $\rho_c = k_c r_c$ for the outgoing scattered wave $O_c(\rho_c)$, the exponential closure depends on two factors: the distance $r_c$ from the nucleus, and how far from the threshold one is $|E - E_{T_c}|$. This is a classical evanescence behavior of quantum tunneling. 

What happens when continuing on the physical sheet $\big\{ E, \hdots, +_c , \hdots \big\}$, as $\left|H_{+}(\rho)\right|$ will now tend to diverge as a ``divide by zero''? The authors have no rigorous answer, but point to the fact that since $E$ is left unchanged by the choice of the $k_c$ sheet, the evanescence result ought to also stand, despite the apparent divergence. 

Note that for photon channels, the semi-classic wavenumber-energy $\rho(E)$ mappings (3) of section II.A of \cite{Ducru_shadow_Brune_Poles_2020} does not yield this behavior, only the relativistic mapping (4) does. 

\end{proof}

We can estimate the orders of magnitude required to experimentally observe this evanescent quantum tunneling closure of the cross sections below threshold.
At distance $r_c$ from the center of mass of the nucleus, and at wavenumber $k_c$, distant from the threshold as $|E-E_{T_c}|$, the asymptotic behavior or the cross-section below threshold is:
\begin{equation}
\begin{IEEEeqnarraybox}[][c]{rcl}
\ln\Big( \sigma_{cc'}(k_c,r_c) \Big) & \ \underset{\begin{array}{c}
   E_c \leq E_{T_c} \\
   k_c \to -\infty
\end{array}  } {\sim}\ & - 2 r_c |k_c| 
\IEEEstrut\end{IEEEeqnarraybox}
\label{eq:sub-threshold cross-section evanescence}
\end{equation}

Assuming a detector is placed at a distance $r_c$ of the nucleus, the cross section would decay exponentially below threshold as the distance $\Delta E_c = | E - E_{T_c}|$ of $E$ to the threshold $E_{T_c}$ increases. For instance, for a threshold of $^{238}$U target reacting with neutron $n$ channel, evanescence (\ref{eq:sub-threshold cross-section evanescence}) would be of the rate of $\log_{10}\Big( \sigma_{cc'}(k_c,r_c) \Big)  \sim  -3\times 10^{16} {r_c}_{\text{m}} \sqrt{{\Delta E_c}_{\text{eV}}} $. 
For a detector placed at a millimeter $r_c \sim 10^{-3} \mathrm{m}$, this means one order of magnitude is lost for the cross section in $\Delta E_c \sim 10^{-27} \text{eV}$, evanescent indeed. Conversely, detecting this quantum tunneling with a detector sensitive to micro-electronvolts $\Delta E_c \sim 10^{-6}\text{eV} \sim 1 \mu \text{eV}$ (200 times more sensitive than the thermal energy of the cosmic microwave background) would see the cross section drop of one order of magnitude for a move of less than $10^{-13}\text{m}$, or a tenth of a pico-meter. We are at sub-atomic level of quantum tunneling: the outgoing wave evanesces into oblivion way before reaching the electron cloud...

Regardless of the evanescence of the transmission matrix, a more general argument on the cross section shows that analytic continuation of the above-threshold expressions will automatically close the channels below the threshold.

\begin{theorem}\label{theo::Analytic continuation annuls sub-threshold cross sections} \textsc{Analytic continuation annuls sub-threshold cross sections}.\\
For massive particles, analytic continuation of above-threshold cross-section expressions to complex wavenumbers $k_c \in \mathbb{C}$ will automatically close channels for real energies $E \in \mathbb{R}$ below thresholds $E - E_{T_c} < 0$
\end{theorem}

\begin{proof}
The proof is based on the fact that massive particles are subject to mappings (2) or (4) section II.A of \cite{Ducru_shadow_Brune_Poles_2020}, which entail wavenumbers are real above threshold, and purely imaginary sub-threshold: $\forall E < E_{T_c} , k_c \in \mathrm{i}\mathbb{R}$.
Let $\psi(\vec{r})$ be a general wave function, so that the probability density is $\left| \psi \right|^2(\vec{r})$.

For a massive particle subject to a real potential, the de Broglie non-relativistic Schr\"odinger equation applies, so that writing the conservation of probability on a control volume, and applying the Green-Ostrogradsky theorem, will yield the following expression for the probability current vector: 
\begin{equation}
\begin{IEEEeqnarraybox}[][c]{rcl}
\vec{j}_\psi  & \; \triangleq \, \frac{\hbar}{\mu}\Im\left[ \psi^* \vec{\nabla} \psi \right]
\IEEEstrut\end{IEEEeqnarraybox}
\label{eq: probability current vector}
\end{equation}
where $\mu$ is the reduced mass of the two-particle system (c.f. equations (2.10) and (2.12) section VIII.2.A, p.312 in \cite{Blatt_and_Weisskopf_Theoretical_Nuclear_Physics_1952}).
By definition, the differential cross section $\frac{\mathrm{d}\sigma_{cc'}}{\mathrm{d}\Omega}$ is the ratio of the outgoing current in channel $c'$ by the incoming current from channel $c$, by unit of solid angle $\mathrm{d}\Omega$.

Consider the incoming channel $c$, classically modeled as a plane wave, $\psi_c(\vec{r}_c) \propto \mathrm{e}^{\mathrm{i}\vec{k}_c\cdot\vec{r}_c}$; and the outgoing channel $c'$, classically modeled as radial wave, $\psi_{c'}(r_{c'}) \propto \frac{ \mathrm{e}^{\mathrm{i}k_{c'} r_{c'}}}{r_{c'}}$.
For arbitrary complex wavenumbers, $k_c , k_{c'}  \in \mathbb{C}$, definition (\ref{eq: probability current vector}) will yield the following probability currents respectively:
\begin{equation}
\begin{IEEEeqnarraybox}[][c]{rcl}
\vec{j}_{\psi_c}  & \propto & \frac{\hbar}{\mu}\Im\left[\mathrm{i}\vec{k}_c \mathrm{e}^{-2\Im\left[\vec{k}_c\right]\cdot\vec{r}_c} \right]   \\ \vec{j}_{\psi_{c'}} & \propto & \frac{\hbar}{\mu}\Im\left[\left(\mathrm{i}k_{c'} - \frac{1}{r_{c'}} \right)\frac{ \mathrm{e}^{-2\Im\left[k_{c'}\right]\cdot r_c}}{r_{c'}^2} \right]\vec{e}_{r}  
\IEEEstrut\end{IEEEeqnarraybox}
\label{eq: probability current vectors for complex k}
\end{equation}
One will note these expressions are not the imaginary part of an analytic function in the wavenumber, because of the imaginary part $\Im\left[k_c\right]$. 
If however we look at real wavenumbers $k_c , k_{c'}  \in \mathbb{R}$, that is at above-threshold energies $E \geq E_{T_c}$, the probability currents (\ref{eq: probability current vectors for complex k}) readily simplify to:
\begin{equation}
\begin{IEEEeqnarraybox}[][c]{rcl}
\vec{j}_{\psi_c}   \propto  \frac{\hbar}{\mu}\Re\left[\vec{k}_c \right]   & \quad , \quad  & \vec{j}_{\psi_{c'}}  \propto  \frac{\hbar}{\mu}\Re\left[k_{c'} \right]\vec{e}_{r}  
\IEEEstrut\end{IEEEeqnarraybox}
\label{eq: probability current vectors for real k}
\end{equation}
These expressions are the real part of analytic functions of the wavenumbers.
If we analytically continue them to complex wavenumbers, and consider the cases of sub-threshold reactions $E < E_{T_c}$, for either the incoming or the outgoing channel, the wavenumbers are then exactly imaginary, $k_c , k_{c'}  \in \mathrm{i}\mathbb{R}$.
The real parts in (\ref{eq: probability current vectors for real k}) become zero, thereby annulling the cross section $\sigma_{c,c'}(E)$.
This means that for massive particles (not massless photons) subject to real potentials, analytic continuation of the probability currents expressions above threshold (\ref{eq: probability current vectors for real k}) will automatically close the sub-threshold channels.
This is true regardless of whether the transmission matrix (\ref{eq:Lane and Thomas Transmission matrix definition}) is or is not evanescent below threshold. 
This constitutes another major argument in favor of analytic continuation of open-channels expressions to describe the closed channels. 
\end{proof}

Note that for photon channels, the derivations for the probability current vector (\ref{eq: probability current vector}) do not stand, and the wavenumber $k_c$ is not imaginary below threshold using mapping (2) nor using the relativistic-correction (4) of section II.A of \cite{Ducru_shadow_Brune_Poles_2020}. 
The fundamental reason why photon treatment is not straightforward is that R-matrix theory was constructed on the semi-classical formalism of quantum physics, with wavefunctions instead of state vectors. Though not incorrect, this wave function approach of quantum mechanics does not translate directly for photons, though some work has been done to describe photons through wave functions \cite{Bialynicki-Birula_Photon_wave_function_1994, Bialynicki-Birula_Photon_wave_function_1996}. 
This is another open area in the field of R-matrix theory, beyond the scope of this article.

\section{\label{sec:Conclusion}Conclusion}

In this article, we conduct a study and establish novel properties of the Siegert-Humblet pole expansion in radioactive states, which we show links R-matrix theory to the Humblet-Rosenfeld pole expansions of the scattering matrix.
The Siegert-Humblet parameters are the poles $\big\{\mathcal{E}_j\big\}$ and residue widths $\left\{r_{j,c}\right\}$ of the Kapur-Peierls $\boldsymbol{R}_{L}$ operator (\ref{eq: def Kapur-Peierls operator}). They are $N_L \geq N_\lambda$ complex, (almost always) simple poles, that reside on the Riemann surface of mapping (\ref{eq:rho_c(E) mapping}), comprised of $2^{N_c}$ branches, and for which one must specify on which sheet they reside, as shown in theorem \ref{theo::Siegert-Humblet Radioactive Poles}. They are intimately interwoven in that not any set of complex parameters is physically acceptable: they must be solution to (\ref{eq:R_L radioactive problem}).
Both $\big\{\mathcal{E}_j\big\}$ and $\big\{r_{j,c}\big\}$ are invariant to changes in boundary conditions $\left\{B_c\right\}$. Furthermore, $\big\{\mathcal{E}_j\big\}$ is invariant to a change in channel radii $\left\{a_c\right\}$, and we established in theorem \ref{theo:: Radioactive parameters transformation under change of channel radius} a simple way of transforming the radioactive widths $\big\{r_{j,c}\big\}$ under a change of channel radius $a_c$.
Since the Siegert-Humblet parameters are the poles and residues of the local Mittag-Leffler expansion (\ref{eq::RL Mittag Leffler}) of the Kapur-Peierls operator $\boldsymbol{R}_{L}$, the set of Siegert-Humblet parameters $\Big\{ E_{T_c}, a_c, \mathcal{E}_j, r_{i,c} \Big\}$ is insufficient to entirely determine the energy behavior of the scattering matrix $\boldsymbol{U}$ through (\ref{eq::u_j scattering residue width}) and (\ref{eq::U Mittag Leffler}). The latter expressions directly link the R-matrix parameters to the poles and residues of the Humblet-Rosenfeld expansion of the scattering matrix, and can be complemented by local coefficients $\left\{\boldsymbol{s}_n\right\}_{\mathcal{W}(E)}$ of the entire part (\ref{eq::Hol_U expansion}), to untangle the energy dependence of the scattering matrix into a simple sum of poles and residues (\ref{eq::U Mittag Leffler}), which is the full Humblet-Rosenfeld expansion of the scattering matrix. Theorem \ref{theo: Scattering matrix poles are the Siegert-Humblet radioactive poles} establishes that under analytic continuation of the R-matrix operators, the poles of the Kapur-Peierls $\boldsymbol{R}_{L}$ operator (i.e. the Siegert-Humblet radioactive poles) are exactly the poles of the scattering matrix $\boldsymbol{U}$.

The latter is one of three results we advance to argue that, contrary to the legacy force-closure of sub-threshold channels presented in Lane \& Thomas \cite{Lane_and_Thomas_1958}, R-matrix operators ought to be analytically continued for complex momenta.
Such analytic continuation is necessary to cancel the spurious poles which would otherwise be introduced by the outgoing wavefunctions, as we establish in theorem \ref{theo: Scattering matrix poles are the Siegert-Humblet radioactive poles}.
Moreover, we show in theorem \ref{theo::satisfied generalized unitarity condition} that the analytic continuation of R-matrix operators in scattering matrix parametrization (\ref{eq:U expression}) enforces Eden \& Taylor's generalized unitarity condition (\ref{eq:: Eden and Taylor Generalized unitarity}). 
Finally, we argue in theorems \ref{theo::evanescence of sub-threshold transmission matrix} and \ref{theo::Analytic continuation annuls sub-threshold cross sections} that analytic continuation will still close cross sections for massive particle channels (not massless photon channels) below threshold.

We thus conclude that the R-matrix community should henceforth come to consensus and agree to set the analytic continuation as the standard way of computing R-matrix operators (in particular the shift $S_c(E)$ and penetration $P_c(E)$ functions) when performing nuclear data evaluations.

\begin{acknowledgments}
This work was partly funded by the Los Alamos National Laboratory (research position in T-2 division during summer 2017), as well as by the Consortium for Advanced Simulation of Light Water Reactors (CASL), an Energy Innovation Hub for Modeling and Simulation of Nuclear Reactors under U.S. Department of Energy Contract No. DE-AC05-00OR22725.

The first author is profoundly grateful to Prof. Semyon Dyatlov, from MIT and U.C. Berkeley, for his critical contribution in pointing to us towards the Gohberg-Sigal theory and providing important insights within it.
He would also like to thank Haile Owusu for his time in discussing Hamiltonian degeneracy; as well as Gr\'egoire Allaire for his important guidance on the Fredholm alternative and the Perron-Frobenius theorem. 
Finally, this work could not have come to fruition without the leadership of Vladimir Sobes at Oak Ridge National Laboratory during the summers of 2015, 2016, and 2018, and to Los Alamos National Laboratory Gerald Hale (author of theorem \ref{theo::satisfied generalized unitarity condition}) and Mark Paris, who organized the R-matrix summer workshops (2016 and 2021), which have been platforms to spark and share these findings.
\end{acknowledgments}

\bibliography{Scattering_matrix_pole_expansions_for_complex_wavenumbers_in_R-matrix_theory}
\end{document}